\documentclass[aps,prl,reprint,superscriptaddress,nofootinbib,longbibliography,floatfix]{revtex4-2}

\usepackage[T1]{fontenc}
\usepackage[utf8]{inputenc}
\usepackage{lmodern}
\usepackage{microtype}
\usepackage{amsmath,amssymb,bm,mathtools}
\usepackage{amsthm}
\usepackage{booktabs}
\usepackage{tabularx}
\usepackage{array}
\usepackage{enumitem}
\usepackage{float}
\usepackage{chngcntr}
\usepackage{graphicx}
\usepackage[colorlinks=true,citecolor=blue,linkcolor=blue,urlcolor=blue]{hyperref}

\newcommand{\Pf}{\operatorname{Pf}}
\newcommand{\polar}{\operatorname{polar}}
\newcommand{\cS}{\mathcal S}
\newcommand{\cR}{\mathcal R}
\newcommand{\TFI}{\mathrm{TFI}}
\newcommand{\OBC}{\mathrm{OBC}}

\newcommand{\diag}{\operatorname{diag}}
\newcommand{\sgn}{\operatorname{sgn}}
\newcommand{\supp}{\operatorname{supp}}

\newcommand{\dd}{\mathrm d}
\newcommand{\e}{\mathrm e}

\newcommand{\ssection}[2]{%
  \refstepcounter{section}%
  \section*{\texorpdfstring{\thesection.\quad #2}{\thesection. #2}}%
  \label{#1}}
\newcommand{\ssubsection}[2]{%
  \refstepcounter{subsection}%
  \subsection*{\texorpdfstring{\thesubsection.\quad #2}{\thesubsection. #2}}%
  \label{#1}}

\newtheorem{theorem}{Theorem}[section]
\newtheorem{proposition}[theorem]{Proposition}
\newtheorem{lemma}[theorem]{Lemma}

\theoremstyle{definition}
\newtheorem{definition}[theorem]{Definition}

\begin{document}

\title{Quantized Stabilizer-R\'enyi Boundary Response across Fermionic SPT Transitions}

\author{M.~A.~Rajabpour}
\affiliation{Instituto de Física, Universidade Federal Fluminense, Av.~Gal.~Milton Tavares de Souza s/n, Gragoatá, 24210-346, Niterói, RJ, Brazil}
\date{\today}

\begin{abstract}
Open boundaries host symmetry-protected Majorana modes, yet their imprint on the stabilizer R\'enyi entropy is obscured by a nonuniversal volume law. At R\'enyi index $\alpha=1/2$, we isolate a bulk-subtracted boundary response in free-fermion BDI chains using exact finite-open-chain relations, large-$L$ Pfaffian evaluations, and direct positive-weight checks. Across a mass inversion, this response approaches $|\Delta\omega|\ln 2$, where $\Delta\omega$ is the change in winding number. The response survives primitive bulk deformations, the tested local boundary perturbations, and broken bulk duality. Thus, for the families studied here, the boundary response counts the Majorana channels that change the fermionic SPT index.
\end{abstract}

\maketitle

\paragraph{Introduction.---}
One-dimensional symmetry-protected topological (SPT) phases are
gapped, short-range-entangled phases that remain distinct under
symmetry-preserving deformations and support protected boundary
degrees of freedom~\cite{Haldane1983,AKLT1987,ChenScience2012,
ChenGuWen2011,Schuch2011,Senthil2015,ChiuRMP2016}. Free-fermion
topological classifications are well established~\cite{Kitaev2001,
Schnyder2008,Kitaev2009,Ryu2010,FidkowskiKitaev2010,FidkowskiKitaev2011}. Since no conventional local order parameter
generally exists, detection rests on nonlocal and entanglement
structures, including string order, characteristic
entanglement-spectrum degeneracies, projective symmetry
representations, strange correlators, and many-body topological
invariants~\cite{denNijsRommelse1989,Pollmann2010,Turner2011,
PollmannTurner2012,You2014,Shapourian2017}. This raises a sharp question: can a global quantity assembled from Pauli expectation values reveal a quantized boundary signature across an SPT transition?

The stabilizer R\'enyi entropy (SRE) turns the Pauli spectrum into a
many-body measure of nonstabilizerness~\cite{Leone2022}. Its extensive
scaling, dynamics, critical behavior, relation to entanglement, and
numerical evaluation have been investigated with tensor networks,
Pauli sampling, Markov chains, Monte Carlo methods, and exact
free-fermion constructions~\cite{Oliviero2022,HaugPiroliMPS2023,
HaugPiroli2023,LamiCollura2023,TarabungaPRXQ2023,Rattacaso2023,
Odavic2023,TarabungaPRL2024,AhmadiGreplova2024,Frau2024,
Passarelli2024,Falcao2025,Ding2025,Collura2026,Hallam2026}.
Wave-function participation entropies provide the natural boundary
framework: Shannon and R\'enyi probabilities in local bases encode
boundary free energies, corner terms, and boundary-condition-changing
operators~\cite{Stephan2009,Stephan2011,AlcarazRajabpour2013,
Stephan2014,AlcarazRajabpour2014,Luitz2014,Tarighi2022}.
A boundary-CFT formulation of critical SREs identifies universal
boundary constants, logarithmic end contributions, and defect-fusion
data~\cite{HoshinoPRX2026,HoshinoPRL2026}. Exact
stabilizer--Shannon correspondences transfer participation-entropy
CFT results to critical quadratic chains~\cite{RamirezRajabpour2025},
while exact polyphase reductions provide lattice realizations beyond
the primitive models~\cite{KhassehRamirezRajabpour2026}. These results
identify the subleading sector as the natural place to seek a
topological response.

For real Gaussian states, the stabilizer--Shannon correspondence maps
the SRE to the computational-basis Shannon--R\'enyi entropy of a
half-filled number-conserving chain on a doubled lattice~\cite{
RamirezRajabpour2025}. Applied to the open TFI family, the doubled
one-body Hamiltonian is the SSH chain~\cite{SuSchriefferHeeger1979},
with the uniform XX chain recovered at criticality. Its topological
boundary structure is also visible in conventional entanglement
diagnostics~\cite{RyuHatsugai2006,Fidkowski2010}. The same boundary
quantity thus has two exact realizations: a Pauli-spectrum response in
the superconducting BDI chain and a computational-basis participation
response in its SSH partner.

Connections between magic and SPT order have also been formulated
through symmetry-protected magic~\cite{Ellison2021}. The total SRE
need not distinguish neighboring one-dimensional SPT
phases~\cite{Catalano2026,LiChang2026}, while the full Pauli spectrum
can retain finer information~\cite{LiChang2026}. Non-Clifford doping and
multipartite subtraction define a distinct topological-magic
response~\cite{Nehra2025}. The object addressed here is the
subleading open-boundary term of the undoped complete chain after
independent removal of the nonuniversal bulk contribution. For
finite-range free-fermion BDI chains, a Laurent symbol $f(z)$ organizes the winding
index, protected Majorana end channels, and the unit-circle zeros
controlling the critical Majorana content~\cite{
VerresenMoessnerPollmann2017,VerresenJonesPollmann2018,
JonesVerresen2019,VerresenThorngren2021,JonesVerresen2023}.

Here we show that this boundary term is quantized. At R\'enyi index
$1/2$, the independently bulk-subtracted difference between matched
massive phases converges to
$\mathcal R_{1/2}^{\rm SPT}=|\Delta\omega|\ln2$. An exactly
shifted-decimated family proves multiplication of the complete
crossover by the number of topology-changing Majorana channels and
independence from the winding common to the two phases. A primitive
deformation outside exact decimation and bulk duality, together with
local termination changes, tests the one-channel value beyond the
solvable construction. Through the Gaussian correspondence, the same
quantization yields a computational-basis Shannon--R\'enyi boundary
difference $2|\Delta\omega|\ln2$ in the associated SSH-type chain.
The response is thus encoded both in the Pauli distribution of the
superconducting BDI chain and in the occupation-probability
distribution of its number-conserving chiral partner.

\paragraph{Finite open-chain formulation.---}
We consider a real BDI chain of length $L$,
\begin{equation}
 H=\frac{i}{2}\sum_{m,n=1}^{L}Z_{mn}\widetilde\gamma_m\gamma_n.
 \label{eq:model}
\end{equation}
Here $Z_{mn}=t_{n-m}$ is the finite Toeplitz matrix associated with the Laurent symbol $f(z)=\sum_a t_a z^a$. For invertible $Z$, the orthogonal matrix entering the Gaussian formulas is the polar factor
$G=\polar(Z)=Z(Z^{\mathsf T}Z)^{-1/2}$. Equation~\eqref{eq:model} covers the finite-range BDI models considered here; their explicit OBC BdG matrices and correlators are collected in the Supplemental Material~\cite{SM}. The $1/2$-SRE is the absolute-minor sum~\cite{RamirezRajabpour2025,KhassehRajabpourFiniteT2026}
\begin{equation}
 \cS_L(G)=\sum_{k=0}^{L}\ \sum_{|I|=|J|=k}|\det G[I,J]|.
 \label{eq:minor-sum}
\end{equation}
The entropy is $M_{1/2}=2\ln\cS_L-L\ln2$. The same minors also have an exact occupation-basis interpretation.  The half-filled ground state of the number-conserving chiral one-body matrix
\begin{equation}
 \mathbb H_Z=
 \begin{pmatrix}
  0&Z\\
  Z^{\mathsf T}&0
 \end{pmatrix}
 \label{eq:chiral-partner}
\end{equation}
has probabilities $p(I,J)=2^{-L}|\det G[I,J]|^2$, and therefore
$M_\alpha(G)=H_\alpha^{\rm ch}(G)$, where $H_\alpha^{\rm ch}$ is the
computational-basis Shannon--R\'enyi entropy.  For
$Z=hI_L+JS_L$, Eq.~\eqref{eq:chiral-partner} is the open SSH chain
with intracell hopping $h$ and intercell hopping $J$.  The derivation,
including the zero-mode convention, is given in Sec.~S2.2 of the
Supplemental Material~\cite{SM}. There are exponentially many terms in Eq.~\eqref{eq:minor-sum}. Interleaving the row and column labels of $G$ defines an antisymmetric checkerboard matrix $R(G)$. In the coherent-minor-sign chamber containing the open transverse-field Ising chain and the shifted-decimated families below, a fixed antisymmetric selector $\mathbb J$ converts every absolute minor to the sign required by the minor-summation formula. We obtain
\begin{equation}
 \cS_L(G)=\left|\Pf\!\left[R(G)+\mathbb J\right]\right|.
 \label{eq:pfaffian}
\end{equation}
The selector, zero-mode convention, coherent-sign chambers, and
limitations of the Pfaffian reduction are detailed in the Supplemental
Material~\cite{SM}. Related minor-summation Pfaffians already solve the
finite-temperature open critical TFI chain~\cite{KhassehRajabpourFiniteT2026};
here they retain both physical ends away from criticality and extract their
$O(1)$ contribution with polynomial cost. The generic anisotropic $XY$
obstruction and the broader OBC model list are also deferred to the
Supplemental Material~\cite{SM}. This Pfaffian condition is logically
distinct from arbitrary-index decimation, which is a tensor-product
statement and need not imply coherent minor signs.

\paragraph{An exact SPT laboratory.---}
The family
\begin{equation}
 f_{r,d}(z;h)=z^r(h+z^d),\qquad r\geq0,\quad d\geq1,
 \label{eq:family}
\end{equation}
separates a background boundary index from the channels that change topology. For $L=r+d\ell$, independent row and column permutations reduce the active open-chain block to $d$ TFI polar factors of size $\ell$, together with $r$ exact zero-mode blocks. For canonical zero-mode occupation sectors this gives the finite-size identity
\begin{equation}
 M_{\alpha}^{(r,d)}(r+d\ell;h)
 =r\ln2+d\,M_{\alpha}^{\TFI}(\ell;h),
 \label{eq:decimation}
\end{equation}
valid for arbitrary $\alpha$. Although Eq.~\eqref{eq:decimation} holds for arbitrary $\alpha$, it does not imply an $\alpha$-independent primitive boundary constant; the boundary-response result below is specific to $\alpha=1/2$. Polyphase reductions for periodic chains and finite intervals are known~\cite{RamirezRajabpour2025,KhassehRamirezRajabpour2026}; for the complete finite OBC chain, the additional ingredients are the active length $L-r$ and the $r$ exact zero-mode blocks. Their derivation, sector convention, and the block Pfaffian at $\alpha=1/2$ are given in the Supplemental Material~\cite{SM}. Equation~\eqref{eq:decimation} supplies exact index multiplication, while the primitive model below tests the one-channel response without tensor-product reduction.

The roots of Eq.~\eqref{eq:family} give, at $h=1$,
\begin{equation}
 (c,\omega_{\mathrm c})=\left(\frac d2,r\right),
 \qquad \omega_{<}=r+d,\qquad \omega_{>}=r,
 \label{eq:labels}
\end{equation}
where $<$ and $>$ denote the sides with the $d$ mobile roots inside and outside the unit disk. Equation~\eqref{eq:decimation} recovers the known open-boundary critical logarithm,
\begin{equation}
 \ln\cS_L^{(r,d)}(1)=s_{\mathrm c}L-\frac{c}{4}\ln L+O(1),
 \label{eq:critical}
\end{equation}
independently of $r$. For the present factorizing boundary class, Eq.~\eqref{eq:critical} reproduces the BCFT logarithm identified in Refs.~\cite{HoshinoPRL2026,RamirezRajabpour2025}.

The $r$ roots pinned at the origin describe localized modes that survive at criticality but do not add scale-free bulk fields. The $d$ roots reaching the unit circle instead create $d$ massless Majorana modes and hence $c=d/2$. Thus the known logarithm is assigned to critical channels, while the common background winding remains in the $O(1)$ boundary sector.

\paragraph{Massive boundary response.---}
For a gapped point in either phase, let
$a\in\{\mathrm{top},\mathrm{triv}\}$ label the two sides of the
transition. We write the large-$L$ expansion as
\begin{equation}
 \ln\cS_L^{a}(m)
 =
 Ls_{a}(m)+C_{a}(m)
 +O\!\left(\mathrm{poly}(L)e^{-mL}\right),
 \label{eq:boundary-expansion}
\end{equation}
where the bulk densities $s_a(m)$ are determined independently and
$m=\xi^{-1}$ is the inverse localization length of the topology-changing root.
The reduced boundary response is then
\begin{equation}
 \begin{aligned}
 \cR_{1/2}(m)
 &\equiv C_{\rm top}(m)-C_{\rm triv}(m),\\
 \cR_{1/2}^{\rm SPT}
 &\equiv \lim_{m\to0^+}\cR_{1/2}(m).
 \end{aligned}
 \label{eq:response}
\end{equation}
Because $M_{1/2}=2\ln\cS_L-L\ln2$, the corresponding difference
between the $O(1)$ terms of $M_{1/2}$ is simply
$2\cR_{1/2}(m)$.

Three regimes are visible before taking the limit in Eq.~\eqref{eq:response}. For $L/\xi\ll1$, the two ends overlap strongly and the state cannot resolve the sign of the mass; the two scaling branches meet at the critical boundary constant. At $L/\xi=O(1)$, the topological branch is nonmonotonic because the emerging Majoranas compete with critical fluctuations, whereas the trivial branch decreases monotonically. For $L/\xi\gg1$, local bulk correlations have saturated and the two edge sectors differ by an $O(1)$ amount, while their remaining finite-size corrections decay with the edge-mode overlap. This is the regime from which $C_{\rm top}$ and $C_{\rm triv}$ are extracted.

The subtraction must be performed at the level of the fitted thermodynamic expansion, not by directly comparing dual-looking finite chains. Exact TFI duality happens to make the bulk densities equal, but a generic BDI deformation generates an analytic odd contribution $L[s_{\rm top}(m)-s_{\rm triv}(m)]$. Such a term can mimic or overwhelm a boundary plateau. We therefore retain the independent bulk-slope difference $\Delta s(m)\equiv s_{\rm top}(m)-s_{\rm triv}(m)$ at each mass and extract the $O(1)$ response only after removing it; equivalent separate-sector fits give the same constants. Local boundary contributions may remain at finite $m$; universality is tested by whether their difference has a common $m\to0^+$ limit.

For the TFI chain, the Pfaffian data show two distinct scaling branches in $L/\xi$. Their difference starts at zero at criticality and approaches $\ln2$ in the massive scaling limit. In the exact family~\eqref{eq:family}, the zero-mode contribution
$r\ln2$ is common to both sides and cancels. Moreover, duality makes
the two bulk densities equal at the paired points $h$ and $h^{-1}$,
so the direct finite-size difference is already bulk-subtracted.
For this exactly dual family, let $L=r+d\ell$ and define the finite-size response by
\begin{equation}
 \Delta_{r,d}^{(L)}(x)
 =
 \ln\cS_{L}^{(r,d)}\!\left(\mathrm e^{-x/\ell}\right)
 -
 \ln\cS_{L}^{(r,d)}\!\left(\mathrm e^{x/\ell}\right),
 \label{eq:Delta-rd}
\end{equation}
where
$x=\ell|\ln h|=(L-r)|\ln h|/d\geq0$.
We denote the corresponding one-channel TFI crossover by
$\Delta_{\TFI}^{(\ell)}(x)\equiv\Delta_{0,1}^{(\ell)}(x)$.
Equation~\eqref{eq:decimation} then gives the exact finite-size identity
\begin{equation}
 \Delta_{r,d}^{(r+d\ell)}(x)
 =
 d\,\Delta_{\TFI}^{(\ell)}(x),
 \qquad
 d=|\omega_<-\omega_>|.
 \label{eq:multiplication}
\end{equation}
Thus the full crossover, not merely its asymptote, counts the channels
that cross the unit circle and ignores the topological background that
survives at criticality. Figure~\ref{fig:exact} displays this exact
index-jump multiplication.

\begin{figure}[t]
    \centering
    \includegraphics[width=\columnwidth]{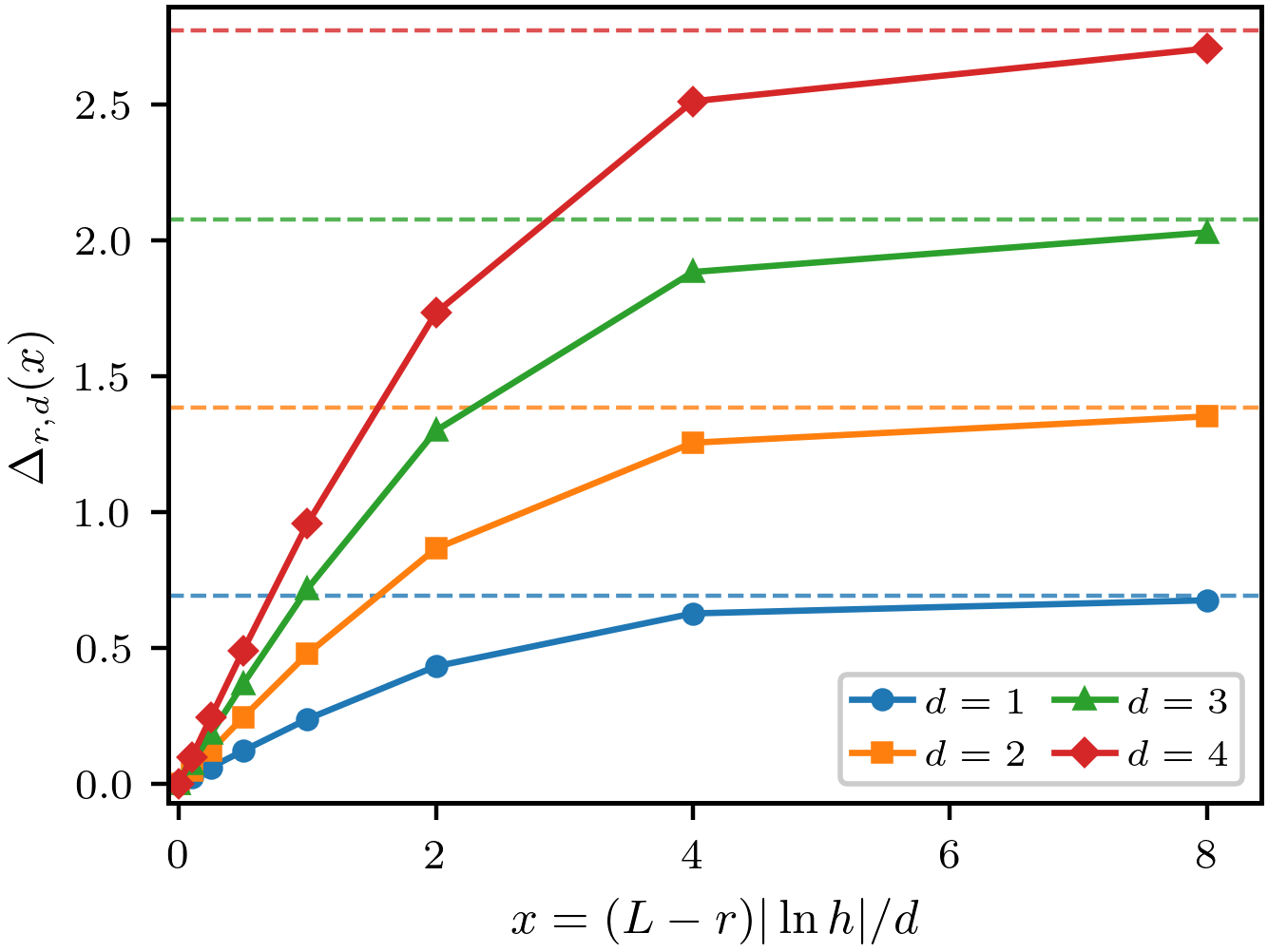}
    \caption{Bulk-subtracted crossover for the shifted-decimated family.
    The response is multiplied by $d=|\Delta\omega|$ and approaches
    $d\ln 2$ (dashed lines). The underlying TFI block size is $\ell=640$.
    }
    \label{fig:exact}
\end{figure}

The multiplication law is exact at every finite compatible size. For example, at the common scaling coordinate $x=8$, primitive TFI block size $\ell=640$ gives $0.676789$ for one channel and precisely $d$ times this value for $d=2,3,4$. Increasing $x$ removes the residual edge overlap and the curves approach $d\ln2$. Varying $r=0,1,2,4$ at fixed $d$ leaves the complete difference unchanged to machine precision even though the individual ground spaces and boundary constants are different. The exact family therefore rules out absolute-index counting: the response is sensitive to the roots that cross, not to zero modes already present on both sides.

The numerical one-channel limit, combined with the exact multiplication law, supports
\begin{equation}
 \cR_{1/2}^{\rm SPT}=|\Delta\omega|\ln2,
 \label{eq:main-result}
\end{equation}
exactly in its multiplication by $|\Delta\omega|$ and numerically in the primitive one-channel constant.

Decimation alone cannot establish universality: in Eq.~\eqref{eq:family}, $d$ is simultaneously the number of TFI copies and $|\Delta\omega|$. Periodically modulated bidiagonal chains furnish nondecimated exact-Pfaffian controls, but remain in the coherent weighted-path class. We therefore use a range-two deformation whose fixed selector demonstrably fails, while also breaking exact replication and bulk duality. Its symbol is
\begin{equation}
 f_{\lambda}(z;h)=h-z^{-1}+\lambda z^{-2},
 \qquad 0\leq\lambda\leq0.2.
 \label{eq:primitive}
\end{equation}
Multiplication by $z^2$ gives a quadratic polynomial. For $\lambda<1/2$, one root crosses the unit circle at $h_{\mathrm c}=1-\lambda$, while the second remains inside; hence $|\Delta\omega|=1$. The support has greatest common divisor one, so no residue-class decomposition produces TFI copies. Moreover, the bulk densities on the two sides are unequal, making the independent subtraction in Eq.~\eqref{eq:response} indispensable. The weighted-path theorem, a trimerized nondecimated control, and the selector failure of Eq.~\eqref{eq:primitive} are detailed in the Supplemental Material~\cite{SM}.

To compare the same infrared mass on the two sides, we set $q_{\rm top}=e^{m}$ and $q_{\rm triv}=e^{-m}$, and choose $h_a=q_a^{-1}-\lambda q_a^{-2}$ for $a\in\{\mathrm{top},\mathrm{triv}\}$. The second root stays inside the unit circle throughout the tested interval. This parameterization fixes the Majorana localization length while allowing the microscopic velocity, bulk density, and irrelevant couplings to vary with $\lambda$. It therefore probes the universality of the boundary response more stringently than moving along a self-dual TFI line.

We additionally perturb only the first and last matrix elements. One test changes the two end onsite terms by unequal amounts; a second changes the first and last active bonds. These deformations leave the thermodynamic symbol, bulk gap, and winding change intact but alter the local boundary scattering problem. A genuinely topological contribution should survive them after smooth boundary terms are removed, whereas a termination-specific Pfaffian constant need not. The same boundary-sector analysis tests the signed-Pfaffian correction for both perturbations; stronger far-from-critical bond perturbations, for which the present analysis is less constraining, are excluded.

At representative massive points and at criticality, an independent-end analysis finds no evidence of an additional edge mode or boundary-level crossing for any of the nine left--right termination pairs; the logarithmic coefficient of $\ln\cS_L$ remains $-1/8$. The Pfaffian mixed-end residual decays exponentially with $L/\xi$, with no absolute-minor correction detected at the available resolution; see the Supplemental Material~\cite{SM}.

Small-$L$ enumeration reveals isolated sign-violating minors in the primitive family. A boundary-pattern decomposition resolves detected logarithmic contributions of order $10^{-23}$ but is not an all-sector upper bound. We therefore also evaluate the exact absolute-minor difference directly, using a positive doubled-Slater bridge at representative points of the strongest deformation. These sign-free estimates agree with the signed-Pfaffian response within a few $10^{-3}$, conservatively $10^{-2}$ after mixing calibration, with no detected growth with $L$~\cite{SM}. The large-$L$ Pfaffian data are fitted at nine masses, $0.00625\leq m\leq0.1$, and seven scaled sizes, $Lm=8,10,\ldots,20$, reaching $L=3200$. A global fit uses the leading overlap correction $e^{-mL}$ together with the near-critical terms $m\ln m$, $m$, $m^2\ln m$, and $m^2$. Across the five bulk deformations and both boundary perturbations, the Pfaffian-assisted endpoints differ from $\ln2$ by at most $7.4\times10^{-6}$; window variation, alternative truncations, and holdout tests give a $2\times10^{-5}$ systematic envelope for this extrapolation. The exact $d$-fold multiplication above is analytic, whereas the primitive one-channel limit is numerical: the $2\times10^{-5}$ envelope applies only to the signed-Pfaffian extrapolation, while direct control of the exact absolute-minor observable remains at the $10^{-3}$--$10^{-2}$ level and excludes an order-unity selector artifact; the logarithmic coefficient remains $-1/8$ throughout, as required for the unchanged $c=1/2$ endpoint~\cite{SM}. The resulting Pfaffian-assisted finite-mass responses and global fits
for the primitive bulk and boundary deformations are shown in
Fig.~\ref{fig:robust}.
\begin{figure*}[t]
 \includegraphics[width=\textwidth]{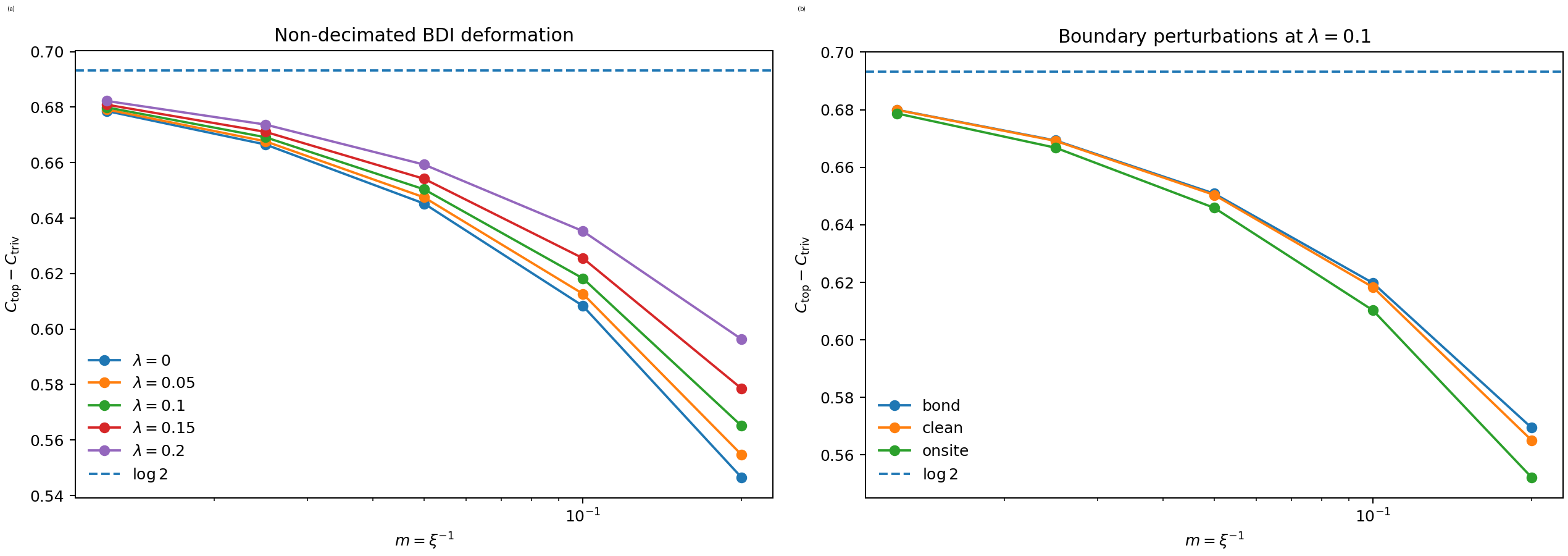}
 \caption{Robustness beyond exact decimation. The independently bulk-subtracted reduced response $\cR_{1/2}(m)$ is shown versus mass. (a) Primitive bulk deformations $f_\lambda$. (b) At $\lambda=0.1$, the unperturbed termination (``clean'') is compared with local perturbations of the end onsite terms (``onsite'') and end bonds (``bond''). Symbols and solid curves are the Pfaffian-assisted finite-mass responses and global fits; the dashed line marks $\ln2$. Direct positive-weight checks of the exact absolute-minor difference are reported in the Supplemental Material.}
 \label{fig:robust}
\end{figure*}
Taken together, the primitive bulk deformation, unequal two-sided bulk densities, local termination changes, the boundary-sector resolution, and the sign-free absolute-minor check exclude decimation, self-duality, a special boundary choice, and an order-unity selector artifact as explanations of the common limit. The remaining discrete input is the number of Majorana channels transferred across the mass inversion.

The large-$L$ analysis must also respect the topological edge singular value. Deep in the ordered phase it is exponentially smaller than machine precision, so an unconstrained singular-value decomposition can choose the zero-mode orientation randomly and introduce a false jump in the Pfaffian. We enforce the finite-chain orientation fixed by $\operatorname{sgn}\det Z$ and verify stability against increased precision. The response is then obtained from a single overdetermined fit to all raw $(L,m)$ values, rather than from two successive nearly saturated extrapolations.

Our quoted uncertainty is systematic rather than a formal least-squares error. Varying the mass and $Lm$ windows, finite-size ansatz, and near-critical truncation, together with holdout tests, gives the conservative envelope $|\cR_{1/2}^{\rm Pf}-\ln2|<2\times10^{-5}$ for the signed-Pfaffian extrapolation of every bulk and boundary deformation; the independent control of the exact absolute-minor observable remains at the $10^{-3}$--$10^{-2}$ level.

\paragraph{Discussion.---}
Equation~\eqref{eq:main-result} identifies a symmetry-protected reduced boundary response, not a topological invariant of the unprocessed total SRE\@. The extensive term is microscopic, and an individual boundary constant $C_a(m)$ can be shifted continuously by changing the termination, adding a decoupled boundary qubit, or selecting another state inside an exact zero-mode manifold; moreover, generic BDI matrices need not lie in a one-Pfaffian sign chamber. Equation~\eqref{eq:response} instead compares two phases using the same microscopic boundary prescription and then approaches the common critical theory. Smooth local contributions cancel in this limiting comparison, while the change in the Majorana boundary sector does not. In this sense the construction resembles other subtractions used to expose universal boundary or topological information, without identifying the full participation entropy with an invariant.

The appearance of $\ln2$ is also consistent with boundary-state counting. A topology-changing BDI channel transfers one real zero mode to each end. The two Majoranas combine into a nonlocal complex fermion with a two-dimensional occupation space, producing a relative factor two. The independent-end analysis is consistent with resolving this factor into effective contributions $\sqrt{2}$ from the two topology-changing Majorana ends, although the present calculation does not identify them with Affleck--Ludwig boundary $g$ factors~\cite{AffleckLudwig1991}.  Our calculation is not a thermodynamic entropy calculation, but the same binary boundary structure is encoded in the absolute-minor partition function. The exact $d$-fold multiplication and the primitive robustness tests are what elevate this interpretation beyond a numerical coincidence in the TFI chain.

This construction is therefore compatible with the vanishing topological magic response reported for the  symmetry-broken and paramagnetic Ising phases~\cite{Nehra2025}: that response probes nonlocally distributed bosonic SPT resources after non-Clifford perturbations, whereas Eq.~\eqref{eq:main-result} probes the fermionic BDI boundary index carried by Jordan--Wigner Majoranas.

The result also extends the BCFT picture of Refs.~\cite{HoshinoPRX2026,HoshinoPRL2026}. At criticality, open ends and conformal defect sectors control logarithmic and constant pieces. A trimerized weighted-path control gives the same one-channel limit within an exact coherent-Pfaffian class, while the range-two deformation is checked directly with positive weights after that coherence fails~\cite{SM}. Our Pfaffian follows the same observable away from the fixed point and finds that a mass inversion leaves a quantized difference between the two boundary free energies. Within the exact family, the BDI data enter separately: $c$ fixes the known critical logarithm, $\omega_{\mathrm c}=r$ labels modes common to both sides, and $|\Delta\omega|=d$ multiplies the massive response. Thus topology is not carried by the dominant volume law, but by the universal boundary remnant that survives after it is removed: in the families studied here, each Majorana channel transferred across the transition leaves a $\ln 2$ imprint on the reduced stabilizer-Rényi response. By Eq.~\eqref{eq:chiral-partner}, the same result gives a computational-basis Shannon--R\'enyi boundary difference $2|\Delta\omega|\ln2$ in the associated half-filled SSH-type chain. The exact Shannon realization also constrains the R\'enyi-index dependence: its largest computational-basis probability is $2^{-L}$, so $M_\infty=L\ln2$ and the full two-sided boundary difference vanishes at $\alpha=\infty$. The critical TFI/SSH participation problem has a boundary R\'enyi transition at $\alpha=4$, but the strict nested-limit massive response at general $\alpha$ remains open; the normalization and exact min-entropy constraint are summarized in the Supplemental Material~\cite{SM}. For models admitting a coherent Pfaffian representation, it would be interesting to determine whether this relation can be derived analytically from the matrix-valued Toeplitz--Hankel structure of the finite-size Pfaffian and the change of factorization index across the transition.

{\it Acknowledgements:}
M.A.R. acknowledges partial support from CNPq and FAPERJ (grant number E-26/210.062/2023). M.A.R. thanks the Abdus Salam International Centre for Theoretical Physics (ICTP) for its hospitality during the completion of this work.

\makeatletter
\let\arxivsavedbibsection\bibsection
\renewcommand{\bibsection}{%
  \begingroup
  \let\addcontentsline\@gobblethree
  \section*{\refname}%
  \endgroup
  \@nobreaktrue}
\makeatother

\makeatletter
\let\bibsection\arxivsavedbibsection
\makeatother

\clearpage
\onecolumngrid
\allowdisplaybreaks[2]
\hypersetup{citecolor=blue,linkcolor=black,urlcolor=blue}

\setcounter{page}{1}
\renewcommand{\thepage}{S\arabic{page}}
\setcounter{section}{0}
\setcounter{subsection}{0}
\setcounter{subsubsection}{0}
\setcounter{equation}{0}
\setcounter{figure}{0}
\setcounter{table}{0}
\setcounter{theorem}{0}
\setcounter{secnumdepth}{3}
\setcounter{tocdepth}{2}
\renewcommand{\thesection}{S\arabic{section}}
\renewcommand{\thesubsection}{\thesection.\arabic{subsection}}
\renewcommand{\thesubsubsection}{\thesubsection.\arabic{subsubsection}}
\makeatletter
\renewcommand{\p@subsection}{}
\renewcommand{\p@subsubsection}{}
\makeatother
\counterwithin{equation}{section}
\counterwithin{figure}{section}
\counterwithin{table}{section}
\renewcommand{\theequation}{\thesection.\arabic{equation}}
\renewcommand{\thefigure}{\thesection.\arabic{figure}}
\renewcommand{\thetable}{\thesection.\arabic{table}}

\begin{center}
{\large\bfseries Supplemental Material for ``Quantized Stabilizer-R\'enyi Boundary Response across Fermionic SPT Transitions''\par}
\vspace{0.9em}
{M.~A.~Rajabpour$^{1}$\par}
\vspace{0.25em}
{\small\itshape $^{1}$Instituto de F\'isica, Universidade Federal Fluminense, Av.~Gal.~Milton Tavares de Souza s/n,\\
Gragoat\'a, 24210-346, Niter\'oi, RJ, Brazil\par}
\vspace{0.3em}
{\small (Dated: \today)\par}
\end{center}

\begin{center}
\begin{minipage}{0.80\textwidth}
\small
We provide the technical foundations and numerical details supporting the results of the Letter. We derive the finite-open-chain covariance matrix for quadratic BDI models, establish the determinant-sum representation of the stabilizer Rényi entropy and its exact Shannon--Rényi realization in the associated number-conserving chiral chain, derive the exact min-entropy constraint on its Rényi-index dependence, and state the precise conditions under which the absolute-minor sum reduces to a single Pfaffian. Exact finite-size reductions are proved for shifted and decimated chains, and a weighted-path theorem gives an exact coherent-Pfaffian class containing periodically modulated nondecimated chains. For the primitive range-two deformation, where exact minor-sign coherence fails, we separate the signed Pfaffian from the absolute-minor sum, resolve boundary-localized noncoherent sectors, and evaluate the complete finite-size absolute-minor difference independently through a positive doubled-Slater bridge at representative parameter points. We also document the topology, boundary perturbations, trimerized control, independent-end analysis, large-system evaluation, and global extrapolation procedure used to extract the boundary response. Finally, complete small-size absolute-minor enumeration for the generic anisotropic $XY$ chain tests both the critical logarithm and the two-sided crossover outside the fixed-selector Pfaffian class.
\end{minipage}
\end{center}
\vspace{0.5em}
\tableofcontents
\vspace{1em}

\noindent
This Supplemental Material is organized to separate two logically different structures used in the Letter.  First, an exact residue-chain reduction may factorize determinant sums at every R\'enyi index.  Second, at index $\alpha=1/2$, an exact single-Pfaffian representation additionally requires a coherent sign pattern for all minors.  We derive the finite-open-chain covariance matrix; prove the exact stabilizer--Shannon correspondence, identify the open TFI partner as the SSH chain, and derive its exact min-entropy constraint; list the model classes and their status; prove the shifted-binomial reduction and a weighted-path coherent-sign theorem; and document the topology and numerical analysis.  For the primitive range-two deformation, for which exact sign coherence fails, the boundary-pattern construction resolves only detected noncoherent sectors and is not used as an upper bound.  The complete absolute-minor difference is checked independently at representative points by a positive doubled-Slater bridge.  A trimerized weighted-path chain supplies a nondecimated exact-Pfaffian control, and complete small-size enumeration treats the generic anisotropic $XY$ chain, where no generic fixed selector is known.  General stabilizer--minor correspondences, reductions in periodic and infinite-system geometries, and the finite-temperature critical Pfaffian are taken from Refs.~\cite{RamirezRajabpour2025,KhassehRamirezRajabpour2026,KhassehRajabpourFiniteT2026}.

\ssection{sec:conventions}{Finite open BDI chains: conventions, covariance matrix, and examples}

\ssubsection{sec:majorana}{Majorana Hamiltonian and Laurent symbol}

Let $\gamma_n$ and $\widetilde\gamma_n$, $n=1,\ldots,L$, be two Majorana species,
\begin{equation}
 \gamma_n^\dagger=\gamma_n,
 \qquad
 \widetilde\gamma_n^\dagger=\widetilde\gamma_n,
\end{equation}
\begin{equation}
 \{\gamma_m,\gamma_n\}=\{\widetilde\gamma_m,\widetilde\gamma_n\}=2\delta_{mn},
 \qquad
 \{\gamma_m,\widetilde\gamma_n\}=0.
 \label{eq:maj-algebra}
\end{equation}
A finite-range translation-invariant BDI chain is specified by real couplings $t_a$ and the Laurent polynomial
\begin{equation}
 f(z)=\sum_{a\in\mathbb Z}t_a z^a,
 \qquad |\supp(f)|<\infty.
 \label{eq:symbol}
\end{equation}
On the interval $1\le n\le L$, open boundary conditions mean that every term whose endpoint leaves the interval is omitted:
\begin{equation}
 H_{\OBC}
 =\frac{i}{2}\sum_{a\in\mathbb Z}t_a
 \sum_{\substack{n=1\\1\le n+a\le L}}^L
 \widetilde\gamma_n\gamma_{n+a}
 =\frac{i}{2}\widetilde\gamma^{\mathsf T}Z^{(L)}\gamma,
 \label{eq:H-obc}
\end{equation}
where
\begin{equation}
 Z^{(L)}_{mn}=t_{n-m}.
 \label{eq:Toeplitz-Z}
\end{equation}
We use the upper open shift
\begin{equation}
 (S_L)_{mn}=\delta_{n,m+1},
 \qquad S_L^L=0,
 \label{eq:shift}
\end{equation}
so that
\begin{equation}
 Z^{(L)}
 =\sum_{a\ge0}t_aS_L^a
 +\sum_{a>0}t_{-a}(S_L^{\mathsf T})^a.
 \label{eq:Z-shift}
\end{equation}

With Pauli operators $\sigma_n^x,\sigma_n^y,\sigma_n^z$ and the Jordan--Wigner convention
\begin{equation}
 \sigma_n^z=i\widetilde\gamma_n\gamma_n,
\end{equation}
the spin Hamiltonian is
\begin{align}
 H_{\rm spin}^{\OBC}
 =&\ \frac{t_0}{2}\sum_{n=1}^{L} \sigma_n^z
 -\sum_{a>0}\frac{t_a}{2}\sum_{n=1}^{L-a}
 \sigma_n^x\!\left(\prod_{j=n+1}^{n+a-1}\sigma_j^z\right)\!\sigma_{n+a}^x
 \nonumber\\
 &-\sum_{a<0}\frac{t_a}{2}\sum_{n=1}^{L-|a|}
 \sigma_n^y\!\left(\prod_{j=n+1}^{n+|a|-1}\sigma_j^z\right)\!\sigma_{n+|a|}^y.
 \label{eq:spin-general}
\end{align}
A transpose of $Z$, an overall sign, or independent signed permutations of its rows and columns leave every absolute-minor sum used below unchanged.  We nevertheless keep one convention throughout each derivation.

\ssubsection{sec:G-derivation}{Derivation of the ground-state matrix \texorpdfstring{$G$}{G}}

Order the Majoranas as
\begin{equation}
 \chi=(\gamma_1,\ldots,\gamma_L,
 \widetilde\gamma_1,\ldots,\widetilde\gamma_L)^{\mathsf T}.
\end{equation}
Equation~\eqref{eq:H-obc} becomes
\begin{equation}
 H_{\OBC}=\frac{i}{4}\chi^{\mathsf T}A\chi,
 \qquad
 A=\begin{pmatrix}
 0&-Z^{\mathsf T}\\
 Z&0
 \end{pmatrix},
 \qquad A^{\mathsf T}=-A.
 \label{eq:A-block}
\end{equation}
We define the covariance matrix by
\begin{equation}
 \Gamma_{ab}=\frac{i}{2}\langle[\chi_a,\chi_b]\rangle
 =\begin{pmatrix}
 0&G^{\mathsf T}\\
 -G&0
 \end{pmatrix}.
 \label{eq:Gamma-def}
\end{equation}
For a gapped quadratic Hamiltonian, the ground-state covariance is the spectral flattening
\begin{equation}
 \Gamma=-A(-A^2)^{-1/2}.
 \label{eq:flattening}
\end{equation}
Since
\begin{equation}
 -A^2=\begin{pmatrix}
 Z^{\mathsf T}Z&0\\
 0&ZZ^{\mathsf T}
 \end{pmatrix},
\end{equation}
comparison of Eqs.~\eqref{eq:Gamma-def} and \eqref{eq:flattening} gives
\begin{equation}
 G=Z(Z^{\mathsf T}Z)^{-1/2}
 =(ZZ^{\mathsf T})^{-1/2}Z
 \equiv\polar(Z).
 \label{eq:G-polar}
\end{equation}
For invertible $Z$, $G$ is orthogonal:
\begin{equation}
 G^{\mathsf T}G=GG^{\mathsf T}=I_L.
 \label{eq:G-orthogonal}
\end{equation}
Equivalently, if
\begin{equation}
 Z=U\Sigma V^{\mathsf T}
 \label{eq:SVD}
\end{equation}
is a singular-value decomposition with $\Sigma>0$, then
\begin{equation}
 G=UV^{\mathsf T}.
 \label{eq:G-SVD}
\end{equation}

If $Z$ has rank $\rho<L$, write
\begin{equation}
 Z=U\begin{pmatrix}\Sigma_\rho&0\\0&0\end{pmatrix}V^{\mathsf T}.
\end{equation}
The partial polar factor fixes the nonzero-energy subspace.  A pure Gaussian ground state additionally requires an orthogonal map $Q\in O(L-\rho)$ between the left and right zero-mode spaces:
\begin{equation}
 G_Q=U\begin{pmatrix}I_\rho&0\\0&Q\end{pmatrix}V^{\mathsf T}.
 \label{eq:G-zero-completion}
\end{equation}
All exact formulas for singular shifted chains below refer to canonical occupation eigenstates, for which $Q$ is a signed permutation matrix.  A generic coherent rotation in a degenerate zero-mode manifold can have a different stabilizer entropy.

\ssubsection{sec:examples}{Explicit models and their \texorpdfstring{$G$}{G} matrices}

The following examples all use Eq.~\eqref{eq:G-polar}; the point is to make explicit which finite matrix is polarized.

\paragraph{Transverse-field Ising/Kitaev chain.}
For
\begin{equation}
 f_{\TFI}(z)=h+Jz,
 \qquad
 Z_{\TFI}^{(L)}=hI_L+JS_L,
 \label{eq:TFI-Z}
\end{equation}
the spin Hamiltonian is
\begin{equation}
 H_{\TFI}
 =\frac{h}{2}\sum_{n=1}^{L}\sigma_n^z
 -\frac{J}{2}\sum_{n=1}^{L-1}\sigma_n^x\sigma_{n+1}^x,
\end{equation}
and
\begin{equation}
 G_{\TFI}^{(L)}(h/J)
 =(hI_L+JS_L)
 \left[(hI_L+JS_L)^{\mathsf T}(hI_L+JS_L)\right]^{-1/2}.
 \label{eq:TFI-G}
\end{equation}
The gauge-equivalent choice $hI_L-JS_L$ is used in the asymptotic formulas of Sec.~\ref{sec:tfi-asymptotics}.

\paragraph{$d$-cluster chain in a field.}
For
\begin{equation}
 f_d(z)=h+Kz^d,
 \qquad
 Z_d^{(L)}=hI_L+KS_L^d,
 \label{eq:cluster-Z}
\end{equation}
the spin interaction is $\sigma_n^x\sigma_{n+1}^z\cdots\sigma_{n+d-1}^z\sigma_{n+d}^x$ and
\begin{equation}
 G_d^{(L)}
 =(hI_L+KS_L^d)
 \left[(hI_L+KS_L^d)^{\mathsf T}(hI_L+KS_L^d)\right]^{-1/2}.
 \label{eq:cluster-G}
\end{equation}
This model admits an exact residue-chain decomposition for every R\'enyi index.

\paragraph{Shifted binomial and topological critical chains.}
For
\begin{equation}
 f_{r,d}(z;h)=z^r(h+z^d),
 \qquad
 Z_{r,d}^{(L)}=S_L^r(hI_L+S_L^d),
 \label{eq:shifted-Z}
\end{equation}
$Z_{r,d}^{(L)}$ has $r$ exact left and $r$ exact right zero modes.  Its pure-state matrix is Eq.~\eqref{eq:G-zero-completion}.  The model $f(z)=z(1+z)$ is the special topological critical case $(r,d,h)=(1,1,1)$.

\paragraph{Anisotropic $XY$ chain.}
For
\begin{equation}
 f_{XY}(z)=h+J_xz+J_yz^{-1},
 \qquad
 Z_{XY}^{(L)}=hI_L+J_xS_L+J_yS_L^{\mathsf T},
 \label{eq:XY-Z}
\end{equation}
\begin{equation}
 G_{XY}^{(L)}=\polar(Z_{XY}^{(L)}).
 \label{eq:XY-G}
\end{equation}
This is a generic BDI Gaussian model, but for $J_xJ_y\ne0$ it is not in the fixed-selector Pfaffian class used in the Letter.

\paragraph{Primitive range-two deformation.}
The non-decimated robustness test uses
\begin{equation}
 f_\lambda(z;h)=h-z^{-1}+\lambda z^{-2},
 \qquad
 Z_\lambda^{(L)}=hI_L-S_L^{\mathsf T}
 +\lambda(S_L^{\mathsf T})^2,
 \label{eq:primitive-Z}
\end{equation}
with
\begin{equation}
 G_\lambda^{(L)}=\polar(Z_\lambda^{(L)}).
 \label{eq:primitive-G}
\end{equation}
It has no residue-chain factorization for $\lambda\ne0$.  Its fixed-selector Pfaffian is not an exact absolute-minor sum at generic finite size; the corresponding logarithmic correction is defined in Sec.~\ref{sec:pf-eval} and evaluated directly in Sec.~\ref{sec:selector-check}.

\ssubsection{sec:model-status}{Which models factorize and which admit the Pfaffian}

The following list is the classification used throughout this work.  An ``all-$\alpha$'' statement is an exact determinant-sum identity, whereas a ``single Pfaffian'' statement concerns only the absolute-minor sum at $\alpha=1/2$.

\begin{enumerate}[leftmargin=2.2em,itemsep=0.35em]
\item \emph{TFI/Kitaev, $f(z)=h+z$.}  This is the primitive building block.  For $h>0$, the fixed selector $J_{2L}$ is exact for the path-oriented representative in Eq.~\eqref{eq:TFI-oriented-representative}, or equivalently a fixed pulled-back selector is exact in the original upper-shift convention.
\item \emph{Connected weighted paths.}  Any invertible bidiagonal $Z$ whose nonzero entries can be gauged to one hopping sign has an exact orientation-adapted coherent selector at $\alpha=1/2$.  Period-two modulation is the minimal nonuniform example; Sec.~\ref{sec:trimerized-control} gives a period-three nondecimated numerical control.
\item \emph{$d$-cluster family, $f(z)=h+z^d$.}  The determinant sums reduce exactly to $d$ TFI residue chains for every $\alpha$; at $\alpha=1/2$ this becomes a direct sum of TFI Pfaffians.
\item \emph{Shifted binomial, $f(z)=z^r(h+z^d)$.}  In canonical zero-mode sectors there is an exact active-block reduction plus $r$ zero-mode blocks; a single block Pfaffian follows at $\alpha=1/2$.
\item \emph{Topological critical chain, $f(z)=z(1+z)$.}  This is a special shifted binomial, and the zero-mode sector must be fixed.
\item \emph{Critical chiral family, $f(z)=z^m+z^{-m}$.}  For compatible sizes, the known equal-amplitude critical reduction applies~\cite{RamirezRajabpour2025}; the Pfaffian follows through that reduction.
\item \emph{Primitive test, $f(z)=h-z^{-1}+\lambda z^{-2}$.}  There is no decimation.  The fixed selector gives a signed-minor Pfaffian rather than an exact absolute-minor sum.  Exact small-size enumeration and the boundary-pattern decomposition resolve detected noncoherent sectors, while Sec.~\ref{sec:absolute-bridge} independently evaluates the complete absolute-minor difference at representative large sizes through a positive-weight bridge.
\item \emph{Generic anisotropic $XY$ chain.}  No arbitrary-index reduction and no generic fixed selector are established when both directional couplings are nonzero.  The Ising endpoints are exceptions.  Nevertheless, Sec.~\ref{sec:xy-exact-check} performs complete small-size enumeration of the exact absolute-minor sum as an independent finite-size test.
\end{enumerate}

\ssection{sec:pfaffian}{Absolute minors and the single-Pfaffian formula}

\ssubsection{sec:entropy-convention}{Determinant sums and entropy convention}

For a real $L\times L$ matrix $G$, define
\begin{equation}
 \mathcal D_\beta(G)
 =\sum_{k=0}^{L}\ 
 \sum_{\substack{I,J\subseteq\{1,\ldots,L\}\\|I|=|J|=k}}
 |\det G[I,J]|^\beta.
 \label{eq:D-beta}
\end{equation}
For the Gaussian states considered here, the Pauli-spectrum R\'enyi entropy in the convention of the Letter is~\cite{RamirezRajabpour2025}
\begin{equation}
 M_\alpha(G)
 =\frac{1}{1-\alpha}
 \ln\!\left[
 \frac{\mathcal D_{2\alpha}(G)}{\mathcal D_2(G)^\alpha}
 \right].
 \label{eq:M-alpha}
\end{equation}
For orthogonal $G$, Cauchy--Binet gives
\begin{equation}
 \mathcal D_2(G)=2^L.
 \label{eq:D2}
\end{equation}
At $\alpha=1/2$,
\begin{equation}
 \cS_L(G):=\mathcal D_1(G),
 \qquad
 M_{1/2}=2\ln\cS_L-L\ln2.
 \label{eq:M-half}
\end{equation}
The resource-normalized quantity is $\widetilde M_\alpha=M_\alpha-L\ln2$.  This subtracts the stabilizer baseline but does not change any two-sided difference at fixed $L$.

The sums $\mathcal D_\beta$ satisfy three elementary identities used repeatedly:
\begin{align}
 \mathcal D_\beta(PGQ)&=\mathcal D_\beta(G),
 \label{eq:D-invariance}\\
 \mathcal D_\beta(G_1\oplus G_2)&=
 \mathcal D_\beta(G_1)\mathcal D_\beta(G_2),
 \label{eq:D-direct-sum}\\
 \mathcal D_\beta(I_r)&=2^r,
 \label{eq:D-identity}
\end{align}
where $P$ and $Q$ are signed permutation matrices.  Equation~\eqref{eq:D-direct-sum} follows because a nonzero minor of a block diagonal matrix selects equal row and column cardinalities inside each block.

\ssubsection{sec:shannon-partner}{Number-conserving chiral partner and the Shannon--R\'enyi entropy}

The determinant sums have an exact interpretation as participation
probabilities of a number-conserving Gaussian state.  Let $Z$ be
invertible and write
\begin{equation}
 Z=U\Sigma V^{\mathsf T},
 \qquad
 G=UV^{\mathsf T}.
 \label{eq:shannon-svd}
\end{equation}
Introduce two sets of complex fermions, $a_i$ and $b_j$, and the
half-filled chiral Hamiltonian
\begin{equation}
 \widehat H_{\rm ch}
 =\sum_{i,j=1}^{L}
 \left(a_i^\dagger Z_{ij}b_j+b_j^\dagger Z_{ij}a_i\right),
 \qquad
 \mathbb H_Z=
 \begin{pmatrix}
  0&Z\\
  Z^{\mathsf T}&0
 \end{pmatrix}.
 \label{eq:chiral-H}
\end{equation}
The negative-energy one-particle orbitals are
\begin{equation}
 w_\mu=\frac{1}{\sqrt2}
 \begin{pmatrix}u_\mu\\-v_\mu\end{pmatrix},
 \qquad E_\mu=-\sigma_\mu,
 \label{eq:negative-orbitals}
\end{equation}
where $u_\mu$ and $v_\mu$ are the columns of $U$ and $V$.  A unitary
rotation inside the occupied subspace changes the Slater determinant
only by an overall phase.  Rotating the occupied orbitals by
$U^{\mathsf T}$ gives
\begin{equation}
 |\Psi_Z\rangle
 =2^{-L/2}\prod_{i=1}^{L}
 \left(a_i^\dagger-\sum_{j=1}^{L}G_{ij}b_j^\dagger\right)|0\rangle,
 \label{eq:slater-G}
\end{equation}
up to an overall phase.

Use the reference configuration with all $a$ sites occupied.  A
half-filled configuration is specified by removing particles from a
set $I$ of $a$ sites and occupying a set $J$ of $b$ sites, with
$|I|=|J|$.  Expansion of Eq.~\eqref{eq:slater-G} gives
\begin{equation}
 \langle I,J|\Psi_Z\rangle
 =\pm 2^{-L/2}\det G[I,J],
 \qquad
 p_Z(I,J)=2^{-L}|\det G[I,J]|^2.
 \label{eq:shannon-probabilities}
\end{equation}
The sign depends only on the ordering convention and drops out of the
probability.  Cauchy--Binet and the orthogonality of $G$ give
$\sum_{I,J}p_Z(I,J)=1$.  Consequently, the computational-basis
Shannon--R\'enyi entropy of the half-filled chiral state is
\begin{align}
 H_\alpha^{\rm ch}(Z)
 &=\frac{1}{1-\alpha}
 \ln\sum_{I,J}p_Z(I,J)^\alpha
 \nonumber\\
 &=\frac{1}{1-\alpha}
 \ln\left[2^{-\alpha L}\mathcal D_{2\alpha}(G)\right]
 =M_\alpha(G),
 \label{eq:stabilizer-shannon}
\end{align}
where the last equality uses $\mathcal D_2(G)=2^L$.

\begin{lemma}[Exact min-entropy constraint]
For every orthogonal $G$ and every finite $L$, the largest probability
of the doubled Slater determinant is
\begin{equation}
 \max_X p_Z(X)=2^{-L}.
 \label{eq:pmax}
\end{equation}
Consequently,
\begin{equation}
 H_\infty^{\rm ch}(Z)=M_\infty(G)=L\ln2.
 \label{eq:min-entropy}
\end{equation}
\end{lemma}
\begin{proof}
The occupied-orbital isometry can be chosen as
\begin{equation}
 W_G=\frac{1}{\sqrt2}
 \begin{pmatrix}G\\ I_L\end{pmatrix},
 \qquad W_G^{\mathsf T}W_G=I_L.
\end{equation}
Every computational-basis amplitude is the determinant of an
$L\times L$ matrix obtained by selecting $L$ rows of $W_G$.  Each row
has Euclidean norm $2^{-1/2}$, so Hadamard's inequality gives
$|\det W_G[X,:]|\leq2^{-L/2}$.  Equality is attained by selecting all
rows from the lower block, for which the amplitude matrix is
$2^{-1/2}I_L$, and also by selecting all rows from the upper block,
for which it is $2^{-1/2}G$.  Squaring proves
Eq.~\eqref{eq:pmax} and hence Eq.~\eqref{eq:min-entropy}.
\end{proof}

At $\alpha=1/2$ this reduces to
\begin{equation}
 H_{1/2}^{\rm ch}
 =2\ln\!\left(2^{-L/2}\cS_L(G)\right)
 =2\ln\cS_L(G)-L\ln2
 =M_{1/2}(G).
 \label{eq:half-shannon}
\end{equation}
If $Z$ is singular, the same construction applies after choosing the
same orthogonal zero-mode completion $G_Q$ as in
Eq.~\eqref{eq:G-zero-completion}; this fixes the corresponding
half-filled state inside the zero-energy manifold.

For the open TFI matrix in Eq.~\eqref{eq:TFI-Z}, Eq.~\eqref{eq:chiral-H}
becomes
\begin{equation}
 \widehat H_{\rm SSH}
 =\sum_{n=1}^{L}h\,a_n^\dagger b_n
 +\sum_{n=1}^{L-1}J\,a_n^\dagger b_{n+1}
 +\mathrm{H.c.}
 \label{eq:SSH-H}
\end{equation}
(up to a staggered gauge).  This is the open SSH chain with intracell
hopping $h$ and intercell hopping $J$; at $h=J$ it is the uniform open
hopping chain, or equivalently the XX chain.  More general finite-range
BDI matrices produce longer-range bipartite SSH-type partners.
Finally, if
\begin{equation}
 \ln\cS_L^a=Ls_a+C_a+o(1),
\end{equation}
then Eq.~\eqref{eq:half-shannon} gives the Shannon--R\'enyi boundary
constant $2C_a$.  The reduced response of the Letter therefore
corresponds to the full occupation-basis boundary difference
\begin{equation}
 \Delta C_{1/2}^{\rm ch}
 =2\cR_{1/2}
 \longrightarrow 2|\Delta\omega|\ln2.
 \label{eq:shannon-response}
\end{equation}

\ssubsection{sec:minor-pf}{Checkerboard embedding and minor-summation identity}

Interleave row and column labels and define the $2L\times2L$ antisymmetric matrix
\begin{equation}
 R(G)_{2m-1,2n}=G_{mn},
 \qquad
 R(G)_{2n,2m-1}=-G_{mn},
 \label{eq:R-G}
\end{equation}
with all odd--odd and even--even entries zero.  Define the selector $J_{2L}$ by
\begin{equation}
 (J_{2L})_{ab}=(-1)^{a+b+1}\quad(a<b),
 \qquad
 (J_{2L})_{ba}=-(J_{2L})_{ab}.
 \label{eq:J-selector}
\end{equation}
For equally sized sets $I,J\subseteq\{1,\ldots,L\}$, let
\begin{equation}
 \epsilon(I,J)=(-1)^{k(k-1)/2+N(I,J)},
 \qquad
 N(I,J)=\#\{(i,j)\in I\times J:j<i\},
 \label{eq:epsilon}
\end{equation}
where $k=|I|=|J|$.

\begin{proposition}[Minor-generating Pfaffian]
For every real matrix $G$ and scalar $u$,
\begin{equation}
 \Pf[J_{2L}+uR(G)]
 =\sum_{k=0}^{L}u^k
 \sum_{|I|=|J|=k}
 \epsilon(I,J)\det G[I,J].
 \label{eq:pf-generating}
\end{equation}
\end{proposition}

\begin{proof}
In the Pfaffian expansion, every factor chosen from $R(G)$ pairs one odd vertex $2i-1$ with one even vertex $2j$.  A term containing $k$ such factors therefore selects two equally sized sets $(I,J)$.  Reordering the selected vertices into odd labels followed by even labels converts their principal Pfaffian into $\epsilon(I,J)\det G[I,J]$.  The complementary vertices are paired by $J_{2L}$.  Direct induction gives unit Pfaffian for every even complementary selector block, while its shuffle sign cancels the sign generated when selected and complementary vertices are separated.  Summing over $(I,J)$ yields Eq.~\eqref{eq:pf-generating}, which is the specialization of the Pfaffian minor-summation formula required here~\cite{IshikawaWakayama1995}.
\end{proof}

\begin{definition}[Coherent selector chamber]
A connected parameter region is coherent with $J_{2L}$ if
\begin{equation}
 \epsilon(I,J)\det G[I,J]\ge0
 \label{eq:coherent}
\end{equation}
for every pair of equally sized sets $(I,J)$ throughout the region.
\end{definition}
In a coherent chamber, Eq.~\eqref{eq:pf-generating} at $u=1$ gives
\begin{equation}
 \cS_L(G)=\Pf[J_{2L}+R(G)].
 \label{eq:single-pf}
\end{equation}
We use an absolute value when the overall orientation has not been fixed.

\ssubsection{sec:path-sign}{The path class and the domain used in the Letter}

\begin{theorem}[Connected weighted-path coherent signs]
Let $Z$ be an invertible bidiagonal matrix with every diagonal and first off-diagonal entry nonzero.  Choose signed permutation matrices $U_Z,V_Z$ such that
\begin{equation}
 \widehat Z=U_Z ZV_Z^{\mathsf T}
 \label{eq:path-oriented-Z}
\end{equation}
is lower bidiagonal with positive diagonal entries and negative subdiagonal entries.  Then
\begin{equation}
 \mathbb H_{\widehat Z}
 =\begin{pmatrix}0&\widehat Z\\ \widehat Z^{\mathsf T}&0\end{pmatrix}
 \label{eq:doubled-path}
\end{equation}
is a connected open path.  With
\begin{equation}
 D_a=\diag(1,-1,1,-1,\ldots),
 \qquad
 D_b=-D_a,
 \label{eq:path-staggered-gauge}
\end{equation}
the one-particle gauge $\mathcal D=\diag(D_a,D_b)$ makes every hopping of $\mathcal D\mathbb H_{\widehat Z}\mathcal D$ nonpositive in the canonical interleaved path ordering.  Define
\begin{equation}
 \widehat G=\polar(\widehat Z)
 =U_ZGV_Z^{\mathsf T},
 \qquad G=\polar(Z).
 \label{eq:path-oriented-G}
\end{equation}
Then the selector-signed minors of $\widehat G$ have one sign,
\begin{equation}
 \epsilon(I,J)\det \widehat G[I,J]\ge0,
 \label{eq:path-coherence}
\end{equation}
up to one common orientation, and therefore
\begin{equation}
 \cS_L(G)=\cS_L(\widehat G)
 =\left|\Pf[J_{2L}+R(\widehat G)]\right|.
 \label{eq:path-single-pf}
\end{equation}
Equivalently, the same result can be written directly in the original row--column convention as
\begin{equation}
 \cS_L(G)=\left|\Pf[R(G)+J_Z]\right|,
 \qquad
 J_Z=\mathcal Q_Z^{\mathsf T}J_{2L}\mathcal Q_Z,
 \label{eq:path-pulled-selector}
\end{equation}
where the signed permutation $\mathcal Q_Z$ acts by $U_Z$ on the odd interleaved coordinates and by $V_Z$ on the even interleaved coordinates.
\end{theorem}
\begin{proof}
Apply the staggered gauge in Eq.~\eqref{eq:path-staggered-gauge}.  In the $L$-particle sector of the canonically ordered $2L$-site path, the resulting nearest-neighbor hopping has nonpositive off-diagonal matrix elements.  Connectivity of the one-particle path makes the fixed-particle configuration graph irreducible.  Perron--Frobenius therefore gives a unique ground-state vector with strictly one-signed amplitudes in the gauged occupation basis.  Undoing the one-particle gauge produces the corresponding Marshall sign in the occupation basis; direct bookkeeping in the row--column convention of Eq.~\eqref{eq:epsilon} shows that this is precisely the configuration sign multiplying the Slater minor.  Filling the $L$ negative-energy orbitals of $\mathbb H_{\widehat Z}$ gives the same state as the doubled Slater determinant of Sec.~\ref{sec:shannon-partner}; its occupation amplitudes are therefore, up to one common orientation, $2^{-L/2}\epsilon(I,J)\det \widehat G[I,J]$.  Hence all selector-signed minors of $\widehat G$ share one sign, and Eq.~\eqref{eq:pf-generating} sums their absolute values.  Signed row and column permutations leave $\cS_L$ invariant, giving the first equality in Eq.~\eqref{eq:path-single-pf}.  Finally,
\begin{equation}
 R(\widehat G)=\mathcal Q_ZR(G)\mathcal Q_Z^{\mathsf T},
\end{equation}
and Pfaffian congruence gives Eq.~\eqref{eq:path-pulled-selector} after taking the absolute value.
\end{proof}

For the upper-shift TFI convention of Eq.~\eqref{eq:TFI-Z}, let $P_L$ reverse the site order and let
\begin{equation}
 D_L=\diag(1,-1,1,-1,\ldots).
\end{equation}
For $J>0$, choosing $U_Z=V_Z=D_LP_L$ gives
\begin{equation}
 \widehat Z_{\TFI}^{(L)}
 =D_LP_L(hI_L+JS_L)P_LD_L
 =hI_L-JS_L^{\mathsf T},
 \qquad
 \widehat G_{\TFI}^{(L)}
 =D_LP_LG_{\TFI}^{(L)}P_LD_L.
 \label{eq:TFI-oriented-representative}
\end{equation}
Thus the open TFI chain belongs to the weighted-path class, but the fixed matrix $J_{2L}$ acts on the path-oriented representative $\widehat G_{\TFI}$; in the original upper-shift convention it is equivalently replaced by the pulled-back selector of Eq.~\eqref{eq:path-pulled-selector}.  For the upper-shift asymptotic convention $Z=hI_L-S_L$, the staggered gauges are unnecessary and one may take $U_Z=V_Z=P_L$.  The lower-bidiagonal convention $Z=hI_L-S_L^{\mathsf T}$ already has $U_Z=V_Z=I_L$.  The theorem also covers arbitrary positive inhomogeneous fields and bonds, including period-two and higher-period modulations and local end changes that keep the path connected.  Exact residue copies and direct sums preserve the same property.

The argument does not extend to the primitive range-two deformation.  With $T=S_L^{\mathsf T}$,
\begin{equation}
 Z_\lambda=hI-T+\lambda T^2
 =h(I-\eta_+T)(I-\eta_-T),
 \qquad
 \eta_\pm=\frac{1\pm\sqrt{1-4\lambda h}}{2h},
 \label{eq:primitive-factorization}
\end{equation}
whenever $1-4\lambda h\ge0$.  After the staggered gauge $D=\diag(1,-1,\ldots)$, both bidiagonal factors in $DZ_\lambda D=h(I+\eta_+T)(I+\eta_-T)$ are totally nonnegative.  This useful property is nevertheless not preserved by the nonlinear polar map.  Moreover, the range-two coupling creates exchange loops in the exterior-power configuration graph, so the path-class Perron--Frobenius proof is obstructed.

The failure can be seen within the retained parameter set.  At $\lambda=0.2$, $m=0.1$, on the topological ($\omega=-1$) side, and $L=7$, the matching prescription in Eq.~\eqref{eq:h-top} gives $h=0.741091267420363\ldots$.  For
\begin{equation}
 I=\{5,6,7\},
 \qquad
 J=\{1,2,3\},
\end{equation}
Eq.~\eqref{eq:epsilon} gives $\epsilon(I,J)=+1$, whereas an $80$-digit evaluation yields
\begin{equation}
 \epsilon(I,J)\det G_\lambda^{(7)}[I,J]
 =-3.312872556893435\ldots\times10^{-10}.
 \label{eq:primitive-negative-minor}
\end{equation}
Thus Eq.~\eqref{eq:single-pf} is not exact for the primitive family.  The relevant question is instead whether the accumulated negative-minor weight changes the logarithm at the scale of the boundary response.  We formulate the exact correction next, resolve its boundary-localized sectors in Sec.~\ref{sec:selector-check}, and evaluate the complete two-sided absolute-minor difference independently in Sec.~\ref{sec:absolute-bridge}.

\ssubsection{sec:pf-eval}{Polynomial evaluation and edge-mode orientation}

In this subsection, $G$ denotes the matrix in the convention used for the Pfaffian evaluation: the path-oriented representative $\widehat G$ for a coherent weighted-path model, and the original lower-oriented matrix for the primitive family below.  Define
\begin{equation}
 K_L=J_{2L}+R(G),
 \qquad
 A_L=|\Pf K_L|,
\end{equation}
define the polynomially accessible signed-minor free energy
\begin{equation}
 \Phi_L:=\ln A_L
 =\frac12\ln|\det K_L|.
 \label{eq:logdet}
\end{equation}
Let $\sigma_L=\sgn(\Pf K_L)$ and
\begin{equation}
 a_{I,J}=\sigma_L\,\epsilon(I,J)\det G[I,J],
 \qquad
 W_-(L)=\sum_{a_{I,J}<0}|a_{I,J}|.
 \label{eq:negative-weight}
\end{equation}
Since $A_L=\sum_{I,J}a_{I,J}$ while $\cS_L=\sum_{I,J}|a_{I,J}|$, one has the exact identities
\begin{align}
 \cS_L&=A_L+2W_-(L),
 \label{eq:absolute-vs-signed}\\
 \ln\cS_L&=\Phi_L+\delta_L,
 \qquad
 \delta_L=\ln\!\left(1+\frac{2W_-(L)}{A_L}\right)\ge0.
 \label{eq:delta-sign}
\end{align}
In a coherent selector chamber $W_-=0$ and Eq.~\eqref{eq:single-pf} is recovered.  A signed logarithmic determinant evaluates $\Phi_L$ in $O(L^3)$ time and $O(L^2)$ memory.  For the primitive data, exact enumeration, the boundary-pattern resolution, and the independent positive-weight calculation below play distinct roles: the first is exact at small size, the second identifies detected noncoherent sectors, and the third checks the complete two-sided absolute-minor difference at representative large sizes.

In a finite topological chain, $Z$ is invertible but its smallest singular value is exponentially small.  Once that value reaches floating-point precision, an unconstrained SVD can return the wrong relative orientation of the two edge singular vectors.  The unique finite-$L$ polar factor satisfies
\begin{equation}
 \det G=\sgn(\det Z).
 \label{eq:det-orientation}
\end{equation}
We restore this orientation by flipping the singular vector associated with the smallest singular value whenever Eq.~\eqref{eq:det-orientation} is violated.  Direct small-$L$ calculations and increased-precision checks show that this removes a numerical completion artifact.

\ssection{sec:decimation}{Exact finite-OBC reductions at arbitrary R\'enyi index}

\ssubsection{sec:unshifted-decimation}{Unshifted polyphase decomposition}

Suppose
\begin{equation}
 f(z)=g(z^d).
 \label{eq:gzd}
\end{equation}
Then $Z_{mn}$ can be nonzero only when $m$ and $n$ have the same residue modulo $d$.  A simultaneous permutation $P_d$ of rows and columns gives
\begin{equation}
 P_dZ^{(L)}P_d^{\mathsf T}
 =\bigoplus_{a=1}^{d}Z_g^{(L_a)},
 \label{eq:unshifted-Z-dec}
\end{equation}
where
\begin{equation}
 L_a=\#\{n\in\{1,\ldots,L\}:n\equiv a\pmod d\},
 \qquad \sum_{a=1}^{d}L_a=L.
\end{equation}
Because the polar decomposition respects orthogonal conjugation and direct sums,
\begin{equation}
 P_dG^{(L)}P_d^{\mathsf T}
 =\bigoplus_{a=1}^{d}G_g^{(L_a)}.
 \label{eq:unshifted-G-dec}
\end{equation}
Equations~\eqref{eq:D-direct-sum} and \eqref{eq:M-alpha} imply
\begin{equation}
 M_\alpha^{g(z^d),\OBC}(L)
 =\sum_{a=1}^{d}M_\alpha^{g,\OBC}(L_a)
 \label{eq:unshifted-M-dec}
\end{equation}
for every $\alpha$.  This is a genuine residue-chain factorization because the same physical-site permutation acts on the two Majorana species.

\ssubsection{sec:shifted-reduction}{Shifted binomials and the active block}

Consider
\begin{equation}
 f_{r,d}(z;h)=z^r(h+z^d),
 \qquad r\ge0,
 \qquad d\ge1,
 \qquad L>r.
 \label{eq:frd}
\end{equation}
Its OBC matrix is
\begin{equation}
 Z_{r,d}^{(L)}=hS_L^r+S_L^{r+d}=S_L^r(hI_L+S_L^d).
 \label{eq:Zrd}
\end{equation}
Set $N=L-r$.

\begin{theorem}[Active-block reduction]
There are independent row and column permutation matrices $P_R,P_C$ such that
\begin{equation}
 P_RZ_{r,d}^{(L)}P_C^{\mathsf T}
 =\begin{pmatrix}
 B_N^{(d)}(h)&0\\
 0&0_r
 \end{pmatrix},
 \qquad
 B_N^{(d)}(h)=hI_N+S_N^d.
 \label{eq:active-Z}
\end{equation}
\end{theorem}

\begin{proof}
A nonzero entry of $Z_{r,d}^{(L)}$ satisfies $n-m=r$ or $n-m=r+d$.  Therefore the last $r$ rows and the first $r$ columns vanish.  Write an active column as $n=r+j$.  For $1\le m,j\le N$,
\begin{equation}
 (Z_{r,d}^{(L)})_{m,r+j}
 =h\delta_{mj}+\delta_{j,m+d},
\end{equation}
which is exactly $hI_N+S_N^d$.
\end{proof}

Since $h>0$, the active block is invertible.  In a canonical zero-mode occupation sector,
\begin{equation}
 P_RG_{r,d}^{(L)}P_C^{\mathsf T}
 =\polar(B_N^{(d)}(h))\oplus Q_r,
 \label{eq:active-G}
\end{equation}
where $Q_r$ is a signed permutation matrix.  Signed row and column permutations reduce $Q_r$ to $I_r$ without changing $\mathcal D_\beta$.

Now group the active indices by residues modulo $d$.  If
\begin{equation}
 N_a=\#\{n\in\{1,\ldots,N\}:n\equiv a\pmod d\},
\end{equation}
then
\begin{equation}
 \Pi_dB_N^{(d)}(h)\Pi_d^{\mathsf T}
 =\bigoplus_{\substack{1\le a\le d\\N_a>0}}
 (hI_{N_a}+S_{N_a}),
 \label{eq:active-decimation-Z}
\end{equation}
so that
\begin{equation}
 G_{r,d}^{(L)}
 \sim_{LR}
 \left[
 \bigoplus_{\substack{1\le a\le d\\N_a>0}}
 G_{\TFI}^{(N_a)}(h)
 \right]\oplus I_r.
 \label{eq:active-decimation-G}
\end{equation}
Here $\sim_{LR}$ denotes equivalence under independent signed row and column permutations.  For $r>0$, Eq.~\eqref{eq:active-decimation-G} is an exact algebraic reduction of determinant sums; it need not be interpreted as a physical qubit permutation.

\ssubsection{sec:shifted-entropy}{Entropy identities and the block Pfaffian}

Equations~\eqref{eq:D-invariance}--\eqref{eq:D-identity} and \eqref{eq:active-decimation-G} give
\begin{equation}
 \mathcal D_\beta(G_{r,d}^{(L)})
 =2^r\prod_{\substack{1\le a\le d\\N_a>0}}
 \mathcal D_\beta(G_{\TFI}^{(N_a)}).
 \label{eq:D-shifted}
\end{equation}
Consequently, for every $\alpha$,
\begin{equation}
 M_\alpha^{(r,d)}(L;h)
 =r\ln2
 +\sum_{\substack{1\le a\le d\\N_a>0}}
 M_\alpha^{\TFI}(N_a;h).
 \label{eq:M-shifted}
\end{equation}
If $L-r=d\ell$,
\begin{equation}
 M_\alpha^{(r,d)}(r+d\ell;h)
 =r\ln2+d\,M_\alpha^{\TFI}(\ell;h).
 \label{eq:M-shifted-equal}
\end{equation}
For the resource-normalized quantity,
\begin{equation}
 \widetilde M_\alpha^{(r,d)}(L;h)
 =\sum_a\widetilde M_\alpha^{\TFI}(N_a;h),
 \label{eq:Mnormalized-shifted}
\end{equation}
because $r+\sum_aN_a=L$.

At $\alpha=1/2$, let $P_N$ reverse the order of $N$ sites, let $D_N=\diag(1,-1,1,-1,\ldots)$, and define the path-oriented TFI representative
\begin{equation}
 \widehat G_{\TFI}^{(N)}(h)
 =D_NP_NG_{\TFI}^{(N)}(h)P_ND_N,
 \qquad
 K_N^{\TFI}(h)=J_{2N}+R(\widehat G_{\TFI}^{(N)}(h)).
 \label{eq:TFI-block-oriented}
\end{equation}
The signed row and column permutations in this definition leave every determinant sum invariant.  Since $2=\Pf(2J_2)$, Eq.~\eqref{eq:D-shifted} becomes one block-diagonal Pfaffian:
\begin{equation}
 \cS_L(G_{r,d}^{(L)})
 =\left|
 \Pf\!\left[
 \left(\bigoplus_aK_{N_a}^{\TFI}(h)\right)
 \oplus(2J_2)^{\oplus r}
 \right]
 \right|.
 \label{eq:shifted-pf}
\end{equation}
Pulling the direct sum back to the original interleaved ordering gives a model-dependent selector $J_{r,d}^{(L)}$ and the equivalent form
\begin{equation}
 \cS_L(G_{r,d}^{(L)})
 =\left|\Pf[R(G_{r,d}^{(L)})+J_{r,d}^{(L)}]\right|.
 \label{eq:shifted-pf-pullback}
\end{equation}

\ssection{sec:topology}{BDI topology and exact consequences for the shifted family}

\ssubsection{sec:index}{Zeros, poles, critical roots, and the BDI index}

Write a Laurent polynomial as
\begin{equation}
 f(z)=z^{q}P(z),
 \qquad P(0)\ne0,
 \qquad q\in\mathbb Z,
\end{equation}
where $q$ is the lowest Laurent power occurring in $f$.  Away from criticality, let $N_{<}$ be the number of zeros of $P$ strictly inside the unit disk, counted with multiplicity.  The BDI index is
\begin{equation}
 \omega=q+N_{<}.
 \label{eq:winding}
\end{equation}
Simple zeros on the unit circle generate massless Majorana channels; if there are $N_{\rm c}$ such zeros, the generic central charge is
\begin{equation}
 c=\frac{N_{\rm c}}{2}.
 \label{eq:c-roots}
\end{equation}
At a critical point, $\omega_{\rm c}$ is computed from zeros remaining strictly inside the disk and from the poles.  This gives the critical label $(c,\omega_{\rm c})$ used for one-dimensional BDI chains~\cite{VerresenJonesPollmann2018,JonesVerresen2019,VerresenThorngren2021,JonesVerresen2023}.

\ssubsection{sec:shifted-labels}{Labels of \texorpdfstring{$z^r(h+z^d)$}{the shifted-binomial symbol}}

The shifted binomial has $r$ zeros at the origin and $d$ roots
\begin{equation}
 z_j=h^{1/d}
 \exp\!\left[\frac{i(\pi+2\pi j)}{d}\right],
 \qquad j=0,\ldots,d-1.
 \label{eq:mobile-roots}
\end{equation}
Therefore
\begin{equation}
 \omega_{<}=r+d\quad(0<h<1),
 \qquad
 \omega_{>}=r\quad(h>1).
 \label{eq:phase-indices}
\end{equation}
At $h=1$, the $d$ mobile roots lie on the unit circle while the $r$ roots at zero remain localized:
\begin{equation}
 (c,\omega_{\rm c})=\left(\frac d2,r\right),
 \qquad
 |\Delta\omega|=d.
 \label{eq:critical-labels}
\end{equation}
The edge localization length associated with a mobile root is
\begin{equation}
 \xi=\frac{d}{|\ln h|}.
 \label{eq:xi-d}
\end{equation}
Thus the natural crossover variable is
\begin{equation}
 x=\frac{L-r}{d}|\ln h|.
 \label{eq:x-shifted}
\end{equation}

\ssubsection{sec:exact-scaling}{Exact separation of \texorpdfstring{$c$}{c}, \texorpdfstring{$\omega_{\rm c}$}{critical omega}, and \texorpdfstring{$\Delta\omega$}{Delta omega}}

For $L=r+d\ell$, Eq.~\eqref{eq:M-shifted-equal} at $\alpha=1/2$ gives
\begin{equation}
 \ln\cS_{L}^{(r,d)}(h)
 =r\ln2+d\ln\cS_{\ell}^{\TFI}(h).
 \label{eq:logS-shifted}
\end{equation}
At criticality, using the open-Ising asymptotic form
\begin{equation}
 \ln\cS_{\ell}^{\TFI}(1)
 =s_{\rm c}\ell-\frac18\ln\ell+C_{\rm c}+o(1),
 \label{eq:TFI-critical-assumed}
\end{equation}
we obtain
\begin{equation}
 \ln\cS_{L}^{(r,d)}(1)
 =s_{\rm c}^{(r,d)}L-\frac{c}{4}\ln L+C_{r,d}+o(1),
 \label{eq:critical-c}
\end{equation}
where $c=d/2$.  The logarithmic coefficient depends on the gapless channels and is independent of the common background index $r$.  Its boundary-CFT origin is established in Refs.~\cite{HoshinoPRX2026,HoshinoPRL2026}; Eq.~\eqref{eq:critical-c} is the exact lattice realization for the factorized BDI family.

Define the finite-size two-sided difference at fixed
$x=\ell|\ln h|=(L-r)|\ln h|/d$ by
\begin{equation}
 \Delta_{r,d}^{(L)}(x)
 =\ln\cS_{L}^{(r,d)}(\e^{-x/\ell})
 -\ln\cS_{L}^{(r,d)}(\e^{x/\ell}).
 \label{eq:Delta-rd-supp}
\end{equation}
We denote the corresponding one-channel finite-size crossover by
\begin{equation}
 \Delta_{\TFI}^{(\ell)}(x)
 \equiv \Delta_{0,1}^{(\ell)}(x).
 \label{eq:Delta-TFI-finite-definition}
\end{equation}
Equation~\eqref{eq:logS-shifted} gives the exact identity
\begin{equation}
 \Delta_{r,d}^{(r+d\ell)}(x)
 =d\,\Delta_{\TFI}^{(\ell)}(x).
 \label{eq:Delta-multiplication}
\end{equation}
Hence the complete crossover is independent of $r=\omega_{\rm c}$ and is multiplied by $d=|\Delta\omega|$.  If the primitive one-channel limit is
\begin{equation}
 \lim_{x\to\infty}\Delta_{\TFI}(x)=\ln2,
 \label{eq:TFI-ln2-conj}
\end{equation}
then the exact family obeys
\begin{equation}
 \lim_{x\to\infty}\Delta_{r,d}(x)
 =|\Delta\omega|\ln2.
 \label{eq:index-response-exact-family}
\end{equation}
Equation~\eqref{eq:Delta-multiplication} is exact; Eq.~\eqref{eq:TFI-ln2-conj} is the numerically supported primitive asymptotic statement.

\ssubsection{sec:chiral-family}{Critical chiral family}

For completeness, the equal-amplitude symbol
\begin{equation}
 f_\nu(z)=z^\nu+z^{-\nu}=z^{-\nu}(z^{2\nu}+1)
\end{equation}
has $2\nu$ unit-circle roots and $\nu$ poles, so
\begin{equation}
 (c,\omega_{\rm c})=(\nu,-\nu).
\end{equation}
For compatible sizes, its known critical reduction gives $2\nu$ Ising building blocks~\cite{RamirezRajabpour2025}.  Thus its critical logarithm is again $-(c/4)\ln L$ in $\ln\cS_L$, while $\omega_{\rm c}$ is nonzero.  This provides a second illustration that the critical logarithm and the surviving boundary index are distinct data.  The unequal-amplitude deformation $az^\nu+bz^{-\nu}$ does not generally inherit this reduction.

\ssection{sec:tfi-asymptotics}{Open TFI chain: bulk, critical, massive, and crossover regimes}

\ssubsection{sec:tfi-bulk}{Bulk density}

For asymptotic evaluation we use the gauge-equivalent TFI matrix
\begin{equation}
 Z_L(h)=hI_L-S_L,
 \qquad h>0.
 \label{eq:TFI-minus}
\end{equation}
The block symbol of the Pfaffian determinant yields
\begin{equation}
 \ln\cS_L(h)=Ls(h)+o(L),
\end{equation}
with
\begin{equation}
 s(h)=\frac{1}{4\pi}\int_0^{2\pi}\ln D_h(\theta)\,\dd\theta,
 \label{eq:s-bulk}
\end{equation}
\begin{equation}
 D_h(\theta)
 =2\left[
 1+\frac{h+1}{\sqrt{1+h^2-2h\cos\theta}}
 \right].
 \label{eq:D-symbol}
\end{equation}
The pointwise identity $D_h(\theta)=D_{1/h}(\theta)$ gives
\begin{equation}
 s(h)=s(1/h).
 \label{eq:s-duality}
\end{equation}
Representative values are
\begin{equation}
 s(1)=0.929695398341610\ldots,
 \qquad
 s(1/2)=s(2)=0.812566878956866\ldots.
\end{equation}
At criticality,
\begin{equation}
 s_{\rm c}=\frac12\ln2+\frac{2G_{\rm Cat}}{\pi},
 \label{eq:sc}
\end{equation}
where $G_{\rm Cat}$ is Catalan's constant.

\ssubsection{sec:tfi-finite}{Critical and massive finite-size terms}

At $h=1$,
\begin{equation}
 \ln\cS_L(1)
 =s_{\rm c}L-\frac18\ln L+C_{\rm c}+O(L^{-1}),
 \label{eq:critical-TFI}
\end{equation}
with
\begin{equation}
 C_{\rm c}=-0.2178154\ldots.
 \label{eq:Cc}
\end{equation}
The logarithmic coefficient agrees with the established open-boundary result~\cite{HoshinoPRL2026,RamirezRajabpour2025}.

For fixed $h$ away from one,
\begin{equation}
 \ln\cS_L(h)=Ls(h)+C_{\TFI}(h)+o(1),
 \label{eq:massive-TFI}
\end{equation}
where $C_{\TFI}(h)$ has distinct branches for $h<1$ and $h>1$.  Although the bulk density is duality symmetric, the physical OBC constants differ.  Numerically,
\begin{center}
\begin{tabular}{c|rrrr}
\toprule
$h$ & $0.5$ & $0.8$ & $1.25$ & $2$\\
\midrule
$C_{\TFI}(h)$ & $0.193262$ & $0.233446$ & $-0.300017$ & $-0.145763$\\
\bottomrule
\end{tabular}
\end{center}
The difference is the first indication that open-chain duality exchanges the boundary termination rather than equating the complete finite system.

\ssubsection{sec:tfi-crossover}{Double-scaling functions}

Set
\begin{equation}
 h_{<}(x,L)=\e^{-x/L},
 \qquad
 h_{>}(x,L)=\e^{x/L},
\end{equation}
where $h_{<}<1$ is the topological branch and $h_{>}>1$ is the trivial branch, and define
\begin{equation}
 \varphi_{\lessgtr}^{(L)}(x)
 =\ln\cS_L\!\left(h_{\lessgtr}(x,L)\right)
 -s_{\rm c}L+\frac18\ln L.
 \label{eq:varphi-finite}
\end{equation}
The data collapse to two distinct functions,
\begin{equation}
 \varphi_{\lessgtr}(x)=\lim_{L\to\infty}\varphi_{\lessgtr}^{(L)}(x).
\end{equation}
Near the critical point, the numerical slopes are
\begin{equation}
 \varphi_{<}(x)=C_{\rm c}+\frac{x}{8}+o(x),
 \qquad
 \varphi_{>}(x)=C_{\rm c}-\frac{x}{8}+o(x).
 \label{eq:small-x-slopes}
\end{equation}
For large $x$, the data are consistent with
\begin{equation}
 \varphi_{\lessgtr}(x)
 =-\frac{x}{4}+\frac18\ln x+\kappa_{\lessgtr}+o(1).
 \label{eq:large-x-varphi}
\end{equation}
The finite-size and scaling-limit differences are
\begin{equation}
 \Delta_{\TFI}^{(L)}(x)
 =\varphi_{<}^{(L)}(x)-\varphi_{>}^{(L)}(x),
 \qquad
 \Delta_{\TFI}(x)=\lim_{L\to\infty}\Delta_{\TFI}^{(L)}(x).
 \label{eq:Delta-TFI}
\end{equation}
It satisfies $\Delta_{\TFI}(0)=0$, has initial slope $1/4$, and approaches $\ln2$ numerically.  At $L=640$,
\begin{center}
\begin{tabular}{c|rrrrrr}
\toprule
$x$ & $0.5$ & $1$ & $2$ & $4$ & $8$\\
\midrule
$\Delta_{\TFI}^{(640)}(x)$
& $0.122943$ & $0.239586$ & $0.433682$ & $0.628262$ & $0.676789$\\
\bottomrule
\end{tabular}
\end{center}
compared with $\ln2=0.693147\ldots$.

\begin{figure}[H]
 \centering
 \includegraphics[width=0.98\linewidth]{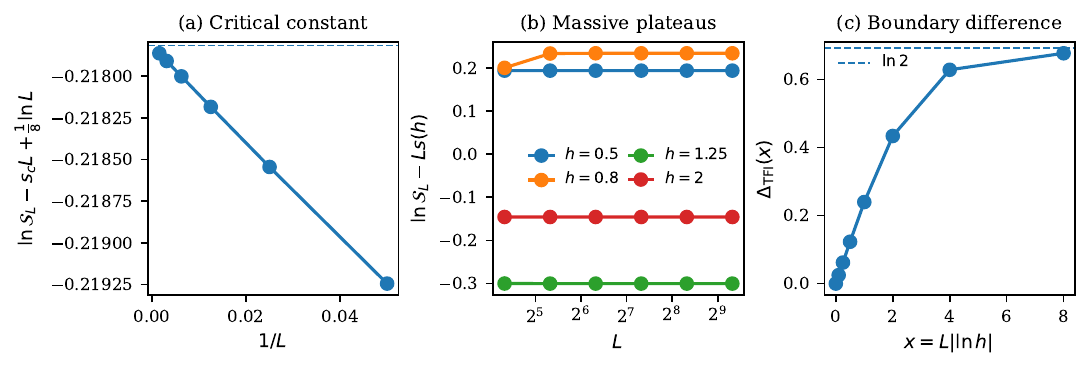}
 \caption{Finite-open-chain TFI scaling.  (a) The critical residual in Eq.~\eqref{eq:critical-TFI}.  (b) Massive boundary constants after subtracting the dual bulk density.  (c) The two-sided crossover difference in Eq.~\eqref{eq:Delta-TFI}; the dashed line is $\ln2$.}
 \label{fig:tfi-scaling}
\end{figure}

\ssection{sec:robustness}{Primitive non-decimated deformation and robustness tests}

\ssubsection{sec:primitive-roots}{A primitive BDI transition with \texorpdfstring{$|\Delta\omega|=1$}{absolute Delta omega equals 1}}

The robustness family is Eq.~\eqref{eq:primitive-Z}.  Multiplying its symbol by $z^2$ gives
\begin{equation}
 p_\lambda(z;h)=hz^2-z+\lambda.
 \label{eq:p-lambda}
\end{equation}
Choose one root to be $q$ and solve
\begin{equation}
 h(q,\lambda)=q^{-1}-\lambda q^{-2}.
 \label{eq:h-q}
\end{equation}
At $q=1$,
\begin{equation}
 h_{\rm c}=1-\lambda.
\end{equation}
For $0<\lambda<1/2$, the second root remains inside the unit disk when $q$ crosses it.  Because the symbol then has a pole of order two,
\begin{equation}
 \omega=-1\quad(q>1),
 \qquad
 \omega=0\quad(q<1),
 \qquad
 |\Delta\omega|=1.
 \label{eq:primitive-indices}
\end{equation}
At the $\lambda=0$ TFI endpoint the pole order reduces to one, with the same winding labels.
We parameterize the two sides by $m=\xi^{-1}>0$:
\begin{align}
 h_{\rm top}(m,\lambda)&=\e^{-m}-\lambda\e^{-2m},
 &q_{\rm top}&=\e^{m},
 \label{eq:h-top}\\
 h_{\rm triv}(m,\lambda)&=\e^{m}-\lambda\e^{2m},
 &q_{\rm triv}&=\e^{-m}.
 \label{eq:h-triv}
\end{align}
The support $\{0,-1,-2\}$ is primitive for $\lambda\ne0$, so this model is not an outer decimation of TFI chains.

At $\lambda=0.1$ we also change the two end onsite terms,
\begin{equation}
 Z_{11}\mapsto Z_{11}+0.35,
 \qquad
 Z_{LL}\mapsto Z_{LL}-0.20,
 \label{eq:onsite-pert}
\end{equation}
and the two end nearest-neighbor matrix elements,
\begin{equation}
 Z_{21}\mapsto Z_{21}+0.20,
 \qquad
 Z_{L,L-1}\mapsto Z_{L,L-1}-0.15.
 \label{eq:bond-pert}
\end{equation}
These are local termination changes and do not alter the thermodynamic roots.  In the boundary-robustness analysis they are applied simultaneously at the two ends, while in the independent-end test of Sec.~\ref{sec:end-factorization} their left and right components are switched independently.

\ssubsection{sec:selector-check}{Test I: direct evaluation of the Pfaffian sign correction}

Complete enumeration for $4\le L\le7$ shows that the small discrepancies previously found between the fixed-selector Pfaffian and the absolute-minor sum are genuine sign corrections, rather than roundoff.  The maximum relative discrepancies over the retained small-size grid are
\begin{center}
\begin{tabular}{c|ccc}
\toprule
Termination & clean & onsite & bond\\
\midrule
Maximum relative correction & $3.63\times10^{-12}$ & $1.12\times10^{-12}$ & $2.07\times10^{-12}$\\
\bottomrule
\end{tabular}
\end{center}
At the explicit point in Eq.~\eqref{eq:primitive-negative-minor}, the complete absolute-minor enumeration gives
\begin{equation}
 \frac{W_-(7)}{A_7}=1.53764020830828\times10^{-12},
 \qquad
 \delta_7=3.07528041661656\times10^{-12}.
 \label{eq:delta-L7}
\end{equation}
A direct all-minor calculation is not possible at the sizes entering the fits.  We therefore integrate out all unrestricted bulk selections exactly and resolve only the boundary sign sectors.

Fix the first and last $w$ physical sites.  A boundary pattern $b$ specifies, for each of the resulting $4w$ row/column vertices, whether the vertex belongs to a minor selected from $R(G)$ or to its complement selected from $J_{2L}$.  Introduce diagonal matrices $P_b^{R}$ and $P_b^{J}$ whose entries on a fixed boundary vertex are $(1,0)$ for a selected vertex and $(0,1)$ for an unselected vertex, while both entries equal one on every unrestricted interior vertex.  Define
\begin{equation}
 K_b=P_b^{J}J_{2L}P_b^{J}+P_b^{R}R(G)P_b^{R}.
 \label{eq:Kb-definition}
\end{equation}
The minor-generating identity implies
\begin{equation}
 P_b:=\Pf K_b,
 \qquad
 \sum_b P_b=\Pf K_L,
 \label{eq:pattern-sum}
\end{equation}
where $P_b$ is the signed sum of all minors with boundary pattern $b$.  Ordering boundary vertices before interior vertices gives
\begin{equation}
 K_b=
 \begin{pmatrix}
 A_b&B_b\\
 -B_b^{\mathsf T}&C
 \end{pmatrix},
 \label{eq:pattern-block}
\end{equation}
with an interior block $C$ independent of $b$.  The interior block $C$ is nonsingular for every retained parameter point and boundary window used below.  The Pfaffian Schur identity then yields
\begin{equation}
 P_b=\Pf C\,
 \Pf\!\left(A_b+B_bC^{-1}B_b^{\mathsf T}\right).
 \label{eq:pf-schur-pattern}
\end{equation}
Thus every bulk minor is summed exactly once, while each boundary pattern requires only a $4w\times4w$ Pfaffian after one bulk factorization.

Normalize the pattern weights by
\begin{equation}
 p_b=\frac{\sigma_LP_b}{A_L},
 \qquad
 \sum_b p_b=1,
 \qquad
 \rho_-^{(w)}=-\sum_{p_b<0}p_b.
 \label{eq:pattern-negative}
\end{equation}
A coarse pattern can contain cancellations between individual minors, so $\rho_-^{(w)}$ is a lower bound on $W_-/A_L$ and increases as the boundary resolution is refined.  The corresponding resolved logarithmic correction is
\begin{equation}
 \delta_L^{(w)}=\ln\!\left(1+2\rho_-^{(w)}\right).
 \label{eq:delta-window}
\end{equation}

We focus on the strongest clean deformation, $\lambda=0.2$.  At $m=0.1$ and $L=120$ on the topological ($\omega=-1$) side, an exhaustive scan of all
\begin{equation}
 2^{4w}=2^{28}=268435456
\end{equation}
patterns at $w=7$ finds $173245$ negative sectors but only
\begin{equation}
 \rho_-^{(7)}=2.348388498863445\times10^{-23},
 \qquad
 \delta_{120}^{(7)}=4.696776997726891\times10^{-23}.
 \label{eq:delta-exhaustive-w7}
\end{equation}
The normalization check gives $\sum_b p_b=1$ within $6\times10^{-15}$.  The $w=8$ and $w=9$ values below are obtained by recursively refining every negative parent sector found at the preceding resolution; they are therefore partial resolved lower bounds rather than exhaustive scans of all $2^{4w}$ patterns.  This refinement gives, respectively,
\begin{equation}
 \delta_{120}^{(8)}=4.7109547676\times10^{-23},
 \qquad
 \delta_{120}^{(9)}=4.7256001733\times10^{-23}.
 \label{eq:delta-window-convergence}
\end{equation}
The increment upon adding one boundary layer is approximately $1.5\times10^{-25}$.  Since a positive coarse pattern can hide cancellations among individual minors, these values quantify only the detected boundary-localized contribution and are not an upper bound on $W_-/A_L$.  As a sensitivity test, independently shifting the two logarithms by $10^{-10}$ leaves every fitted endpoint unchanged on the displayed scale; this exercise is not used as a bound on the physical correction.  The independent full-observable check is given in Sec.~\ref{sec:absolute-bridge}.

For completeness, the boundary-resolved corrections at the smallest size $L=12/m$ of each massive fit for $\lambda=0.2$ are
\begin{center}
\begin{tabular}{c|c|cc|c}
\toprule
$m$ & $L$ & $\delta_{\rm top}^{(w)}$ & $\delta_{\rm triv}^{(w)}$ & $\delta_{\rm top}^{(w)}-\delta_{\rm triv}^{(w)}$\\
\midrule
$0.1000$ & $120$ & $4.73\times10^{-23}$ & $1.75\times10^{-24}$ & $4.55\times10^{-23}$\\
$0.0500$ & $240$ & $6.20\times10^{-23}$ & $8.81\times10^{-24}$ & $5.32\times10^{-23}$\\
$0.0250$ & $480$ & $5.70\times10^{-23}$ & $1.88\times10^{-23}$ & $3.82\times10^{-23}$\\
$0.0125$ & $960$ & $4.96\times10^{-23}$ & $2.74\times10^{-23}$ & $2.22\times10^{-23}$\\
\bottomrule
\end{tabular}
\end{center}
At fixed $m=0.1$, the six-site resolved value is unchanged at the $10^{-25}$ level between $L=120$ and $L=240$, showing that this resolved contribution is an $O(1)$ boundary correction rather than a hidden bulk density.  The end-deformed cases are still smaller: for $\lambda=0.1$, $m=0.1$, and $L=120$, exhaustive $w=6$ scans find no negative sector on the topological side for either termination, while the trivial-side resolved corrections are $1.26\times10^{-51}$ and $1.76\times10^{-51}$ for the onsite and bond changes, respectively.

The resolved sectors are therefore extremely small, but their smallness cannot by itself certify the unobserved sectors.  The signed-Pfaffian extrapolation below is consequently reported separately from the positive-weight comparison of the complete absolute-minor response.

\ssubsection{sec:absolute-bridge}{Test II: positive doubled-Slater evaluation of the exact absolute-minor difference}

The doubled-Slater probabilities of Sec.~\ref{sec:shannon-partner} provide a sign-free estimator of the actual observable.  At matched topological and trivial points, write
\begin{equation}
 p_a(X)=2^{-L}|\det A_a(X)|^2,
 \qquad a\in\{{\rm top},{\rm triv}\},
 \label{eq:bridge-probabilities}
\end{equation}
where $A_a(X)$ is the $L\times L$ selected-row matrix of the unnormalized block $\bigl( G_a^{\mathsf T},I_L\bigr)^{\mathsf T}$; the corresponding occupied-orbital isometry is $2^{-1/2}A_a(X)$.  Define the positive bridge
\begin{equation}
 Z_t=\sum_X p_{\rm triv}(X)^{(1-t)/2}p_{\rm top}(X)^{t/2},
 \qquad 0\le t\le1.
 \label{eq:positive-bridge}
\end{equation}
Since $Z_0=2^{-L/2}\cS_L^{\rm triv}$ and $Z_1=2^{-L/2}\cS_L^{\rm top}$,
\begin{equation}
 \Delta_L^{\rm abs}(m)
 \equiv\ln\cS_L^{\rm top}-\ln\cS_L^{\rm triv}
 =\frac12\int_0^1\!\dd t\,
 \left\langle\ln p_{\rm top}-\ln p_{\rm triv}\right\rangle_t,
 \label{eq:bridge-identity}
\end{equation}
where the expectation uses the normalized weight in Eq.~\eqref{eq:positive-bridge}.  No selector sign or Pfaffian approximation enters this identity.

A configuration is an $L$-row subset of the $2L\times L$ block $\bigl( G_a^{\mathsf T},I_L\bigr)^{\mathsf T}$.  If the row in slot $r$ is replaced by a candidate row $v^{\mathsf T}$, the determinant ratio on side $a$ is
\begin{equation}
 d_a=v^{\mathsf T}A_a^{-1}e_r.
 \label{eq:bridge-det-ratio}
\end{equation}
A heat-bath exchange update therefore assigns conditional weight
\begin{equation}
 w(v|r)\propto |d_{\rm triv}|^{1-t}|d_{\rm top}|^t.
 \label{eq:bridge-heatbath}
\end{equation}
The inverse is updated by a rank-one formula and periodically recomputed.  We use eight-point Gauss--Legendre quadrature in $t$ and several independent chains from inequivalent initial bases; the uncertainties below are based on run-to-run dispersion.

Complete enumeration at $\lambda=0.2$, $m=0.1$, and $L=7$ gives
\begin{equation}
 \Delta_7^{\rm abs}=0.0652225288794837,
 \qquad
 \Delta_7^{\rm Pf}=0.0652225288778876,
 \label{eq:bridge-exact-L7}
\end{equation}
so the exact two-sided sign correction is $1.5961\times10^{-12}$.  At $L=6$ and $m=0.3$, the bridge estimate $0.12310(48)$ agrees with the exact value $0.12296268$.  Table~\ref{tab:sign-free-bridge} gives the large-size comparison.  The reported quantity is the raw two-sided difference at identical $(L,m)$; its extensive bulk term is immaterial because it cancels in the comparison between the two evaluation methods.

\begin{table}[H]
\centering
\caption{Direct positive-weight evaluation of the primitive $\lambda=0.2$ absolute-minor difference.  Parentheses denote between-run uncertainty.}
\label{tab:sign-free-bridge}
\begin{tabular}{ccccc}
\toprule
$m$ & $L$ & $\Delta_L^{\rm Pf}$ & $\Delta_L^{\rm abs}$ & $\Delta_L^{\rm abs}-\Delta_L^{\rm Pf}$\\
\midrule
$0.10$ & $40$  & $ 0.1725669$ & $ 0.17374(13)$ & $ 0.00117(13)$\\
$0.10$ & $60$  & $ 0.0147862$ & $ 0.01267(27)$ & $-0.00211(27)$\\
$0.10$ & $80$  & $-0.1807670$ & $-0.17900(36)$ & $ 0.00177(36)$\\
$0.10$ & $100$ & $-0.3830776$ & $-0.37983(38)$ & $ 0.00325(38)$\\
$0.10$ & $120$ & $-0.5864641$ & $-0.58289(53)$ & $ 0.00358(53)$\\
$0.05$ & $120$ & $ 0.0360746$ & $ 0.03460(33)$ & $-0.00148(33)$\\
\bottomrule
\end{tabular}
\end{table}

The offsets are statistically resolved relative to the finite-run dispersions quoted in the table and therefore should not be interpreted as pointwise compatible with zero.  Their absolute magnitudes nevertheless remain below $3.6\times10^{-3}$, and the available points do not establish a growth law with $L$.  A constant diagnostic fit gives $0.00074\pm0.00101$, but no constancy in $L$ or $m$ is assumed.  Short-chain calibrations in the exactly coherent TFI class show that finite-run dispersions can underestimate slow-mixing systematics.  We therefore attach a conservative mixing scale of order $10^{-2}$ and use the bridge only to conclude that no order-unity selector artifact is detected, with present direct control at the $10^{-3}$--$10^{-2}$ level; it does not convert the boundary-pattern values into a rigorous upper bound.

As a second check, sample the escort distribution
\begin{equation}
 q_a(X)=\frac{|\det A_a(X)|}{\cS_L^a},
 \qquad
 s_a(X)=\sgn\!\left[\sigma_L\epsilon(I,J)\det G_a[I,J]\right].
\end{equation}
Then
\begin{equation}
 \delta_L^a=-\ln\langle s_a\rangle_{q_a}.
 \label{eq:escort-sign}
\end{equation}
At $\lambda=0.2$, $m=0.1$, and $L=120$, no negative-selector configuration was observed among $13336$ one-sweep-spaced samples on either side; the heat-bath move rates were $0.907$ and $0.889$.  Under an independence assumption this would give a $95\%$ limit $\delta_L^a<4.5\times10^{-4}$ per side.  Counting only one effective sample per ten sweeps weakens the corresponding two-sided scale to about $9\times10^{-3}$.  This is a statistical rare-event check, not a mathematical bound.

\ssubsection{sec:critical-check}{Test III: the critical endpoint without decimation}

At $m=0$, each $\lambda$ has one simple unit-circle root and $c=1/2$.  We fit the polynomially evaluated signed-minor free energy
\begin{equation}
 \Phi_L
 =s_{\rm c}(\lambda)L
 +\beta_{\log}(\lambda)\ln L
 +C(\lambda)+\frac{a(\lambda)}{L}.
 \label{eq:critical-fit}
\end{equation}
The results are
\begin{center}
\begin{tabular}{c|rrrrr}
\toprule
$\lambda$ & $0$ & $0.05$ & $0.10$ & $0.15$ & $0.20$\\
\midrule
$\beta_{\log}$
& $-0.124999$ & $-0.125000$ & $-0.125001$ & $-0.125003$ & $-0.125006$\\
\bottomrule
\end{tabular}
\end{center}
At $\lambda=0.1$, the onsite and bond terminations give $-0.124996$ and $-0.124999$.  The sign correction is again invisible on this scale: at the strongest deformation $\lambda=0.2$, exhaustive $w=6$ boundary scans give $\delta_{120}^{(6)}=3.63\times10^{-23}$ and $\delta_{240}^{(6)}=3.77\times10^{-23}$, while no negative $w=6$ sector is resolved at $L=120$ for either $\lambda=0.1$ boundary termination.  The known open-Ising endpoint is therefore retained after both primitive bulk and local boundary deformations.

\ssubsection{sec:bulk-subtraction}{Test IV: independent bulk subtraction}

For each fixed choice of $\lambda$ and termination, define the raw two-sided quantity
\begin{equation}
 \Delta\Phi_L(m)
 =\Phi_L\!\left(h_{\rm top}(m,\lambda)\right)
 -\Phi_L\!\left(h_{\rm triv}(m,\lambda)\right),
 \label{eq:raw-response}
\end{equation}
where the fixed deformation and termination labels are suppressed.  Because the two bulk densities are unequal for $\lambda\ne0$, this difference contains a nonuniversal term linear in $L$.  Defining
$\Delta s(m)\equiv s_{\rm top}(m)-s_{\rm triv}(m)$
for that fixed dataset, we remove it directly at fixed $m$ through
\begin{equation}
 \Delta\Phi_L(m)
 =\Delta s(m)L+\cR_{1/2}(m)
 +A_{\rm ov}(m)\e^{-mL}.
 \label{eq:massive-fit}
\end{equation}
The exponential correction is fixed by the matched localization scale $|q_{\rm top}|^{-L}=|q_{\rm triv}|^{L}=\e^{-mL}$.  Equation~\eqref{eq:massive-fit} is equivalent to fitting the two sectors separately and subtracting their constants, but avoids subtracting two independently estimated $O(1)$ terms.

We evaluate nine masses,
\begin{equation}
 m=0.1,\ 0.075,\ 0.05,\ 0.0375,\ 0.025,\ 0.01875,
 \ 0.0125,\ 0.009375,\ 0.00625,
 \label{eq:mass-grid}
\end{equation}
and seven target scaled sizes $Lm\simeq8,10,12,14,16,18,20$, choosing for each target the nearest integer chain length and reaching $L=3200$.  The central fixed-mass fits use the targets $Lm\simeq12,14,16,18,20$; the two smaller target scaled sizes are retained for window-stability tests.  The largest root-mean-square residual over all clean and boundary-deformed signed-Pfaffian fits is $7.1\times10^{-8}$.  The separate positive-weight comparison in Sec.~\ref{sec:absolute-bridge} determines the current direct control of the signed-to-absolute difference.

\ssubsection{sec:nearcritical}{Test V: near-critical extrapolation and boundary robustness}

To avoid two successive extrapolations, the central analysis fits all raw values in Eqs.~\eqref{eq:raw-response} and \eqref{eq:mass-grid} simultaneously:
\begin{align}
 \Delta\Phi_L(m)
 ={}&\Delta s(m)L+A_{\rm ov}(m)\e^{-mL}
 \notag\\
 &+\cR_{1/2}^{\rm SPT}
 +u_1m\ln m+u_2m+u_3m^2\ln m+u_4m^2.
 \label{eq:response-extrapolation}
\end{align}
For each model or termination, $\Delta s(m)$ and $A_{\rm ov}(m)$ are independent nuisance parameters at every mass, whereas $\cR_{1/2}^{\rm SPT},u_1,u_2,u_3,u_4$ are common to the nine masses.  The central $Lm\ge12$ fit therefore contains $45$ raw values and $23$ parameters, leaving $22$ residual degrees of freedom.

The resulting one-channel endpoints are
\begin{table}[H]
\centering
\caption{Near-critical one-channel response from the signed-Pfaffian global fit.  The quoted $2\times10^{-5}$ variation is a conservative systematic envelope obtained by changing the scaled-size window, mass window, and retained subleading terms; the independent absolute-minor check has the lower precision stated in Sec.~\ref{sec:absolute-bridge}.}
\label{tab:plateaus}
\begin{tabular}{c|c|c}
\toprule
Test & $\cR_{1/2}^{\rm SPT}$ & $\cR_{1/2}^{\rm SPT}-\ln2$\\
\midrule
$\lambda=0$ & $0.693145509$ & $-1.67\times10^{-6}$\\
$\lambda=0.05$ & $0.693145947$ & $-1.23\times10^{-6}$\\
$\lambda=0.10$ & $0.693146777$ & $-4.03\times10^{-7}$\\
$\lambda=0.15$ & $0.693148341$ & $+1.16\times10^{-6}$\\
$\lambda=0.20$ & $0.693151356$ & $+4.18\times10^{-6}$\\
Bond termination & $0.693139811$ & $-7.37\times10^{-6}$\\
Onsite termination & $0.693140707$ & $-6.47\times10^{-6}$\\
\bottomrule
\end{tabular}
\end{table}
All seven endpoints lie within $7.4\times10^{-6}$ of $\ln2$.  The systematic ensemble uses $Lm_{\min}=10,12,14$, the full mass range and the restricted windows $m\le0.075$ and $m\le0.05$, and nested near-critical truncations obtained by omitting one or both of the $m^2\ln m$ and $m^2$ terms; a second exponential is used only as a conditioning check.  A simultaneous fit imposing one common endpoint gives $\cR_{1/2}^{\rm SPT}=0.693145493$ for $Lm\ge12$ and $\cR_{1/2}^{\rm SPT}=0.693146658$ for $Lm\ge14$, respectively $1.69\times10^{-6}$ and $5.23\times10^{-7}$ below $\ln2$.  Fits that omit the smallest mass predict it with errors below $5.5\times10^{-7}$; omitting the two smallest masses gives a maximum holdout error $1.02\times10^{-6}$.  Across the controlled signed-Pfaffian fit variants,
\begin{equation}
 |\cR_{1/2}^{\rm Pf}-\ln2|<2\times10^{-5}.
 \label{eq:response-envelope}
\end{equation}
Combined with the direct positive-weight agreement in Sec.~\ref{sec:absolute-bridge}, the data support
\begin{equation}
 \lim_{m\to0^+}
 [C_{\rm top}(m,\lambda)-C_{\rm triv}(m,\lambda)]
 =\ln2
 \label{eq:primitive-ln2}
\end{equation}
for one primitive index-changing channel.  The $2\times10^{-5}$ number characterizes the smooth Pfaffian-assisted extrapolation; the present independent control of the exact absolute-minor observable is at the $10^{-3}$--$10^{-2}$ level.

\begin{figure}[H]
 \centering
 \includegraphics[width=0.98\linewidth]{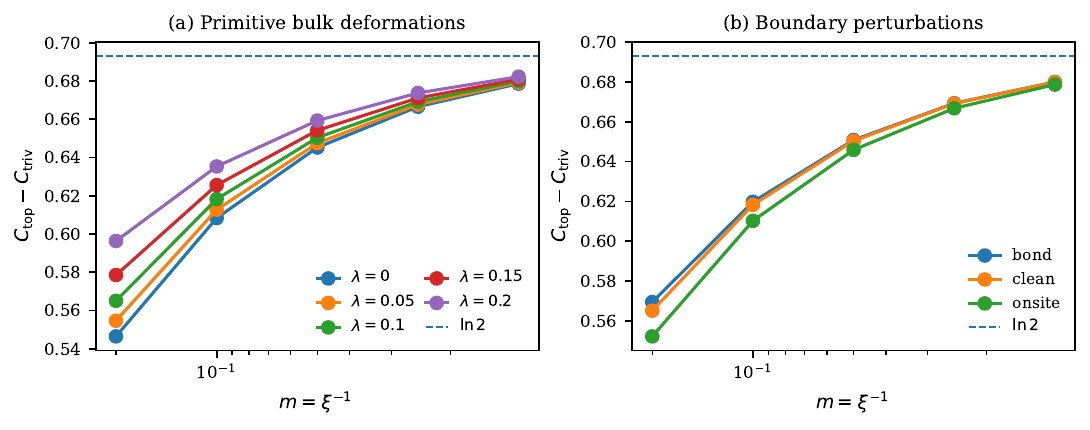}
 \caption{Bulk-subtracted massive response for the primitive family, showing representative masses from the original and enlarged grids.  (a) Bulk range-two deformations destroy exact TFI replication and bulk duality.  (b) Local end onsite and end-bond changes alter finite-mass constants.  The global analysis uses all nine masses in Eq.~\eqref{eq:mass-grid}.}
 \label{fig:robustness}
\end{figure}

\begin{figure}[H]
 \centering
 \includegraphics[width=0.92\linewidth]{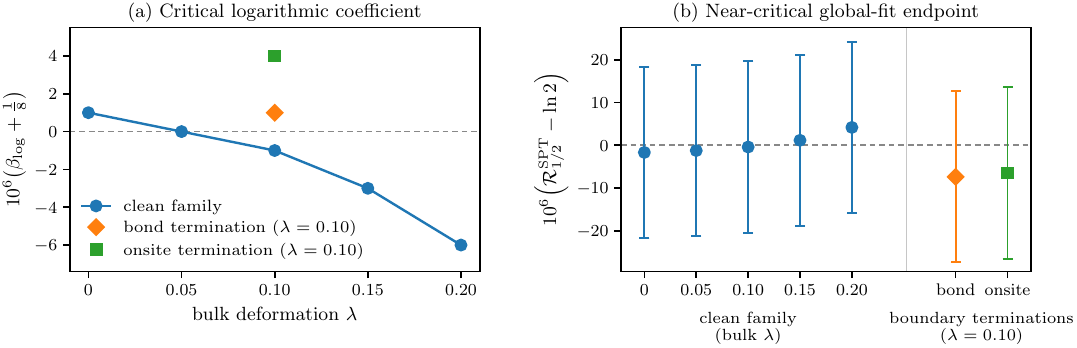}
 \caption{Numerical consistency checks. (a) Fitted critical logarithmic coefficient
relative to $-1/8$. The clean family is shown as a function of the bulk
deformation $\lambda$; both boundary-termination tests are performed at
$\lambda=0.10$. (b) Direct global-fit endpoints relative to $\ln 2$.
The error bars show the conservative $2\times10^{-5}$ variation of the signed-Pfaffian extrapolation in Eq.~\eqref{eq:response-envelope}; the direct absolute-minor comparison is given separately in Table~\ref{tab:sign-free-bridge}.}
 \label{fig:systematics}
\end{figure}

\ssubsection{sec:end-factorization}{Test VI: independent-end factorization and boundary-spectrum analysis}

We finally test whether the massive $O(1)$ contribution separates into independent left- and right-end terms.  Let
\begin{equation}
 Z_{ab}^{(L)}
 =Z_{\lambda}^{(L)}+\Delta_L^{(a)}+\Delta_R^{(b)},
 \qquad
 a,b\in\{0,{\rm on},{\rm bond}\},
 \label{eq:independent-end-Z}
\end{equation}
where $0$ denotes the clean termination, ${\rm on}$ selects the corresponding left or right term in Eq.~\eqref{eq:onsite-pert}, and ${\rm bond}$ selects the corresponding term in Eq.~\eqref{eq:bond-pert}.  These finite-rank changes preserve the real bipartite BDI form and leave the bulk Laurent symbol and winding unchanged.

We first test whether the chosen terminations introduce an additional boundary mode at representative massive points and at criticality.  Because $Z_{ab}^{(L)}$ remains lower triangular,
\begin{equation}
 \det Z_{ab}^{(L)}
 =h^{L-2}(h+a_L)(h+a_R),
 \label{eq:boundary-det-audit}
\end{equation}
where $a_L\in\{0,0.35\}$ and $a_R\in\{0,-0.20\}$.  None of these factors vanishes on the retained grid.  More stringently, for $\lambda=0.1$, $m=0.1$, and $L=120$, the topological sector has exactly one exponentially small singular value,
\begin{equation}
 8.56\times10^{-7}\le \sigma_1\le9.57\times10^{-7},
 \qquad
 8.005\times10^{-2}\le \sigma_2\le8.011\times10^{-2},
 \label{eq:boundary-singular-top}
\end{equation}
whereas the trivial sector obeys
\begin{equation}
 8.539\times10^{-2}\le \sigma_1\le8.543\times10^{-2}.
 \label{eq:boundary-singular-triv}
\end{equation}
The ranges are over all nine independent termination pairs.  The singular vectors associated with the small topological singular value remain localized at opposite ends.  Continuously multiplying all four boundary changes by a common factor $0\le t\le2$ leaves $\sigma_2^{\mathrm{top}}>0.0799$ and $\sigma_1^{\mathrm{triv}}>0.0853$, so no boundary-level crossing connects the clean and perturbed terminations.  At the critical point, fits of all nine combinations to Eq.~\eqref{eq:critical-fit} give
\begin{equation}
 -0.125049\le\beta_{\log}\le-0.124940,
 \label{eq:boundary-beta-audit}
\end{equation}
showing that the open-Ising boundary logarithm is unchanged within the numerical accuracy.

Let $\tau\in\{\mathrm{top},\mathrm{triv}\}$ label the phase.  For the exact absolute-minor free energy, write
\begin{equation}
 F_{ab}^{\tau}(L,m)
 \coloneqq\ln\cS_L(G_{ab}^{\tau})
 =\Phi_{ab}^{\tau}(L,m)+\delta_{ab}^{\tau}(L,m),
 \label{eq:independent-end-F}
\end{equation}
with $\Phi$ and $\delta$ defined in Eqs.~\eqref{eq:logdet} and \eqref{eq:delta-sign}.  The mixed-end residual is
\begin{equation}
 X_{ab}^{\tau}
 =F_{ab}^{\tau}-F_{a0}^{\tau}
  -F_{0b}^{\tau}+F_{00}^{\tau}.
 \label{eq:mixed-end-residual}
\end{equation}
This combination cancels the extensive bulk contribution and every term depending on only one end.  Hence an asymptotic form
\begin{equation}
 F_{ab}^{\tau}(L,m)
 =Ls_{\tau}(m)+c_{L,\tau}(a,m)+c_{R,\tau}(b,m)+o(1)
 \label{eq:end-additive-form}
\end{equation}
implies $X_{ab}^{\tau}\to0$ for $L/\xi\to\infty$.

The polynomially evaluated part $X_{ab}^{\Phi,\tau}$ exhibits a clean exponential decay.  Representative data at $\lambda=0.1$ and $m=0.1$ are
\begin{center}
\begin{tabular}{c|cc}
\toprule
$Lm$ & $\max_{a,b}|X_{ab}^{\Phi,\mathrm{top}}|$ & $\max_{a,b}|X_{ab}^{\Phi,\mathrm{triv}}|$\\
\midrule
$8$  & $3.90\times10^{-6}$ & $8.98\times10^{-9}$\\
$12$ & $9.17\times10^{-8}$ & $8.71\times10^{-11}$\\
$16$ & $1.98\times10^{-9}$ & $1.02\times10^{-12}$\\
\bottomrule
\end{tabular}
\end{center}
where the maxima are over $a,b\in\{{\rm on},{\rm bond}\}$.  The same decay is found for $m=0.05$ and $0.025$, with the topological sector converging more slowly, as expected from the overlap of its two edge modes.

The signed-to-absolute correction is evaluated independently.  Exact all-minor enumeration through $L=8$ gives
\begin{equation}
 \max_{a,b}\left|X_{ab}^{\tau}-X_{ab}^{\Phi,\tau}\right|
 <1.4\times10^{-14}.
 \label{eq:mixed-end-exact-audit}
\end{equation}
At $m=0.1$ and $L=120$, the boundary-pattern decomposition resolves no negative sector for any of the nine combinations at $w=4$, nor for any of the four nontrivial mixed pairs at $w=5$.  Thus no correction to the mixed-end conclusion is detected at the available boundary resolution.  This is not promoted to an all-sector upper bound.

These tests support, for the nonsingular symmetry-preserving terminations considered here,
\begin{equation}
 \cS_{L,ab}^{\tau}
 \sim \e^{Ls_{\tau}}g_{L,\tau}(a)g_{R,\tau}(b),
 \label{eq:effective-end-factorization}
\end{equation}
as an effective massive large-$L/\xi$ factorization.  The quantities $g_{L,\tau}$ and $g_{R,\tau}$ are effective boundary factors: the numerical result neither identifies them with Affleck--Ludwig $g$ factors nor extends Eq.~\eqref{eq:effective-end-factorization} to arbitrary boundary perturbations.

The tests remove the main possible artifacts: the $\lambda$ term destroys decimation, an independent bulk-slope parameter at every mass removes broken bulk duality, Eqs.~\eqref{eq:onsite-pert} and \eqref{eq:bond-pert} alter microscopic boundary scattering, the boundary-pattern calculation locates extremely small noncoherent sectors, the positive bridge checks the complete absolute-minor difference at the strongest clean deformation, and the independent-end analysis finds no evidence that the retained terminations introduce an additional boundary mode or change the critical boundary universality class.  The Pfaffian mixed-end residual vanishes exponentially with $L/\xi$, and no absolute-minor correction to this conclusion is detected at the available resolution.  The enlarged grid, global Pfaffian fit, window variations, and near-critical holdout tests give the systematic envelope in Eq.~\eqref{eq:response-envelope}.  Together with the exact multiplication in Eq.~\eqref{eq:Delta-multiplication}, they support
\begin{equation}
 \cR_{1/2}^{\rm SPT}=|\Delta\omega|\ln2
 \label{eq:main-response}
\end{equation}
as a near-critical boundary response.  Equation~\eqref{eq:main-response} is an exact consequence of the primitive value within the shifted-decimated family, while the primitive value itself remains a numerically supported universality statement.

\ssubsection{sec:trimerized-control}{Test VII: periodically modulated weighted-path control}

The weighted-path theorem gives a nondecimated exact-Pfaffian control that is logically distinct from the range-two test.  Period-two modulation is the minimal example.  We use a period-three chain to avoid any appearance of a disguised uniform dimerization and to test three inequivalent microscopic terminations.  In a lower-bidiagonal convention,
\begin{equation}
 Z_{nn}^{(q)}(m)=h_n^{(q)}(m),
 \qquad
 Z_{n+1,n}^{(q)}(m)=-J_n^{(q)},
 \label{eq:trimer-Z}
\end{equation}
with periodic critical patterns
\begin{align}
 (J_1,J_2,J_3)&=(1.70,0.65,1.25),\\
 (h_1^{\rm c},h_2^{\rm c},h_3^{\rm c})
 &=(0.6145752337021859,2.011337128479881,1.1174095158221562),
 \label{eq:trimer-parameters}
\end{align}
and $h_a(m)=\e^m h_a^{\rm c}$.  The termination $q=0,1,2$ cyclically shifts both patterns.  The products obey
\begin{equation}
 h_1^{\rm c}h_2^{\rm c}h_3^{\rm c}=J_1J_2J_3=1.38125.
\end{equation}
For a three-site cell, the bulk chiral block is
\begin{equation}
 Z(k;m)=
 \begin{pmatrix}
 h_1(m)&0&-J_3\e^{-ik}\\
 -J_1&h_2(m)&0\\
 0&-J_2&h_3(m)
 \end{pmatrix},
\end{equation}
with
\begin{equation}
 \det Z(k;m)=J_1J_2J_3\bigl(\e^{3m}-\e^{-ik}\bigr).
 \label{eq:trimer-det}
\end{equation}
Thus $m<0$ and $m>0$ differ by $|\Delta\omega|=1$, the critical point has one massless Majorana channel, and the zero-mode recurrence gives $\xi=|m|^{-1}$.  The chain is one connected path rather than a residue-class direct sum, its one-site translation symmetry and ordinary TFI self-duality are absent, and Eq.~\eqref{eq:single-pf} is exact by the weighted-path theorem.

At criticality, fits of
\begin{equation}
 \ln\cS_L=s_{\rm c}L+\beta_{\log}\ln L+C_q+O(L^{-1})
\end{equation}
for compatible lengths up to $L=1440$ give the coefficients in Table~\ref{tab:trimerized-control}.  In the massive phases, define
\begin{equation}
 D_L^{(q)}(m)=\ln\cS_L^{(q)}(-m)-\ln\cS_L^{(q)}(+m).
\end{equation}
The bulk slopes are unequal; for example $\Delta s(0.1)=5.97408\times10^{-4}$ and $\Delta s(0.025)=1.51984\times10^{-4}$.  We therefore fit
\begin{align}
 D_L^{(q)}(m)={}&L\Delta s(m)+D_{0,q}
 +a_qm\ln m+b_qm+c_qm^2\ln m+d_qm^2
 \notag\\
 &+A_q(m)\e^{-mL},
 \label{eq:trimer-global-fit}
\end{align}
using the same nine masses as Eq.~\eqref{eq:mass-grid}, scaled sizes $mL=8,10,12,14,16,18$, and compatible lengths up to $L=2880$.

\begin{table}[H]
\centering
\caption{Critical logarithm and independently bulk-subtracted endpoint for the three trimerized terminations.}
\label{tab:trimerized-control}
\begin{tabular}{cccc}
\toprule
$q$ & $\beta_{\log}$ & $D_{0,q}$ & $D_{0,q}-\ln2$\\
\midrule
$0$ & $-0.124999999805$ & $0.693144955817$ & $-2.225\times10^{-6}$\\
$1$ & $-0.124999999934$ & $0.693147745734$ & $+5.652\times10^{-7}$\\
$2$ & $-0.125000000247$ & $0.693145765407$ & $-1.415\times10^{-6}$\\
\bottomrule
\end{tabular}
\end{table}

All three terminations reproduce $\beta_{\log}=-1/8$.  Varying the upper mass window over $m_{\max}=0.1,0.075,0.05,0.0375$ gives
\begin{equation}
 \max_{q,m_{\max}}|D_{0,q}-\ln2|=7.39\times10^{-6}.
 \label{eq:trimer-envelope}
\end{equation}
This control shows that nondecimation, enlarged-unit-cell inhomogeneity, unequal two-sided bulk densities, and termination dependence are compatible with the same one-channel response inside an exactly coherent Pfaffian class.  It does not replace the range-two test: the latter is retained precisely because its selector fails and the true absolute-minor difference must be checked independently.

\ssubsection{sec:xy-exact-check}{Test VIII: complete absolute-minor checks for the generic anisotropic \texorpdfstring{$XY$}{XY} chain}

The generic anisotropic $XY$ chain provides a test outside the
fixed-selector Pfaffian chamber.  We set
\begin{equation}
 J_x=\frac{1+\gamma}{2},
 \qquad
 J_y=\frac{1-\gamma}{2},
 \qquad
 Z_{XY}^{(L)}=hI_L+J_xS_L+J_yS_L^{\mathsf T},
 \label{eq:xy-check-Z}
\end{equation}
and evaluate $G_{XY}^{(L)}=\polar(Z_{XY}^{(L)})$.  For every retained
point we compute
\begin{equation}
 \cS_L(G)=
 \sum_{k=0}^L
 \sum_{\substack{|I|=|J|=k}}
 |\det G[I,J]|
 \label{eq:xy-check-exact-D1}
\end{equation}
by complete numerical enumeration of all minors.  No Pfaffian selector
or minor-sign assumption enters this calculation.

The Laurent symbol is
\begin{equation}
 f_{XY}(z)=h+J_xz+J_yz^{-1},
\end{equation}
so its polynomial part is
\begin{equation}
 P_{XY}(z)=J_xz^2+hz+J_y.
 \label{eq:xy-polynomial}
\end{equation}
At $h=1$,
\begin{equation}
 P_{XY}(z)=(z+1)(J_xz+J_y).
 \label{eq:xy-critical-factorization}
\end{equation}
For every $0<\gamma\le1$, one root is the simple unit-circle root
$z=-1$, whereas the second root is
\begin{equation}
 z_2=-\frac{J_y}{J_x}=-\frac{1-\gamma}{1+\gamma},
 \qquad |z_2|<1.
 \label{eq:xy-second-root}
\end{equation}
Thus the critical theory contains one massless Majorana channel and has
$c=1/2$.  The limit $\gamma\to0$ is not included in this statement,
because the second root then also approaches the unit circle.

\paragraph{Critical logarithm.}
At $h=1$, define
\begin{equation}
 \beta_{\rm eff}(L,\gamma)
 =
 -L^2\left[
 \ln\cS_{L+1}-2\ln\cS_L+\ln\cS_{L-1}
 \right].
 \label{eq:xy-beta-eff-improved}
\end{equation}
For
$\ln\cS_L=sL+\beta_{\log}\ln L+C+\cdots$,
this estimator approaches $\beta_{\log}$.  We extend the complete
critical enumeration to $L=14$ for
$\gamma=0.4,0.6,0.8,1$ and to $L=13$ for
$\gamma=0.5,0.7,0.9$.

To avoid mistaking the sizeable small-$\gamma$ corrections for a change
of exponent, we fit the local estimator rather than the full
volume-law-dominated quantity:
\begin{equation}
 \beta_{\rm eff}(L,\gamma)
 =
 \beta_\infty+
 \sum_{p=1}^P\frac{a_p(\gamma)}{L^p}.
 \label{eq:xy-beta-joint-fit}
\end{equation}
A joint free-exponent fit to
$\gamma=0.7,0.8,0.9,1$ gives
\begin{center}
\begin{tabular}{c|c|cc}
\toprule
$P$ & $L_{\min}$ & $\beta_\infty$ & RMS residual\\
\midrule
3 & 6 & -0.124770973 & $2.20\times10^{-6}$ \\
3 & 7 & -0.124771340 & $1.69\times10^{-6}$ \\
3 & 8 & -0.124849242 & $6.51\times10^{-7}$ \\
4 & 6 & -0.124790963 & $1.10\times10^{-6}$ \\
4 & 7 & -0.125011133 & $1.18\times10^{-7}$ \\
4 & 8 & -0.125058502 & $1.41\times10^{-8}$ \\
\bottomrule
\end{tabular}
\end{center}
The six controlled variants have mean
\begin{equation}
 \beta_\infty=-0.124875,
 \qquad
 \max|\beta_\infty+1/8|=2.3\times10^{-4}.
 \label{eq:xy-beta-final}
\end{equation}
Thus the freely fitted exponent for the better-controlled anisotropies
agrees with $-1/8$.

The smaller anisotropies have much larger correction amplitudes and do
not independently determine the asymptote at these sizes.  As a
consistency test, we fix $\beta_{\log}=-1/8$ and fit
\begin{equation}
 \ln\cS_L
 =
 s(\gamma)L-\frac18\ln L+C(\gamma)
 +\sum_{p=1}^4\frac{b_p(\gamma)}{L^p},
 \qquad L\ge5.
 \label{eq:xy-critical-constrained}
\end{equation}
The residuals are
\begin{center}
\begin{tabular}{c|ccc}
\toprule
$\gamma$ & $L_{\max}$ & RMS residual & maximum residual\\
\midrule
0.4 & 14 & $1.24\times10^{-6}$ & $2.62\times10^{-6}$ \\
0.5 & 13 & $4.31\times10^{-8}$ & $8.07\times10^{-8}$ \\
0.6 & 14 & $8.09\times10^{-8}$ & $1.24\times10^{-7}$ \\
0.7 & 13 & $1.13\times10^{-8}$ & $1.97\times10^{-8}$ \\
0.8 & 14 & $3.68\times10^{-10}$ & $7.10\times10^{-10}$ \\
0.9 & 13 & $1.94\times10^{-10}$ & $3.18\times10^{-10}$ \\
1 & 14 & $3.71\times10^{-11}$ & $5.76\times10^{-11}$ \\
\bottomrule
\end{tabular}
\end{center}
Even at $\gamma=0.4$, where the raw local estimator is visibly
pre-asymptotic, the fixed-$-1/8$ form describes every retained size with
an RMS error $1.3\times10^{-6}$.  Taken together, the controlled
free-exponent fits at larger anisotropy and the constrained fits over
the full retained range are consistent with
\begin{equation}
 \ln\cS_L(1,\gamma)
 =
 s(\gamma)L-\frac18\ln L+O(1),
 \qquad
 M_{1/2}
 =
 a_{1/2}(\gamma)L-\frac14\ln L+O(1).
 \label{eq:xy-critical-final}
\end{equation}
For the smaller anisotropies this is a consistency statement rather
than an independent asymptotic extraction.

\begin{figure}[H]
 \centering
 \includegraphics[width=0.82\linewidth]
 {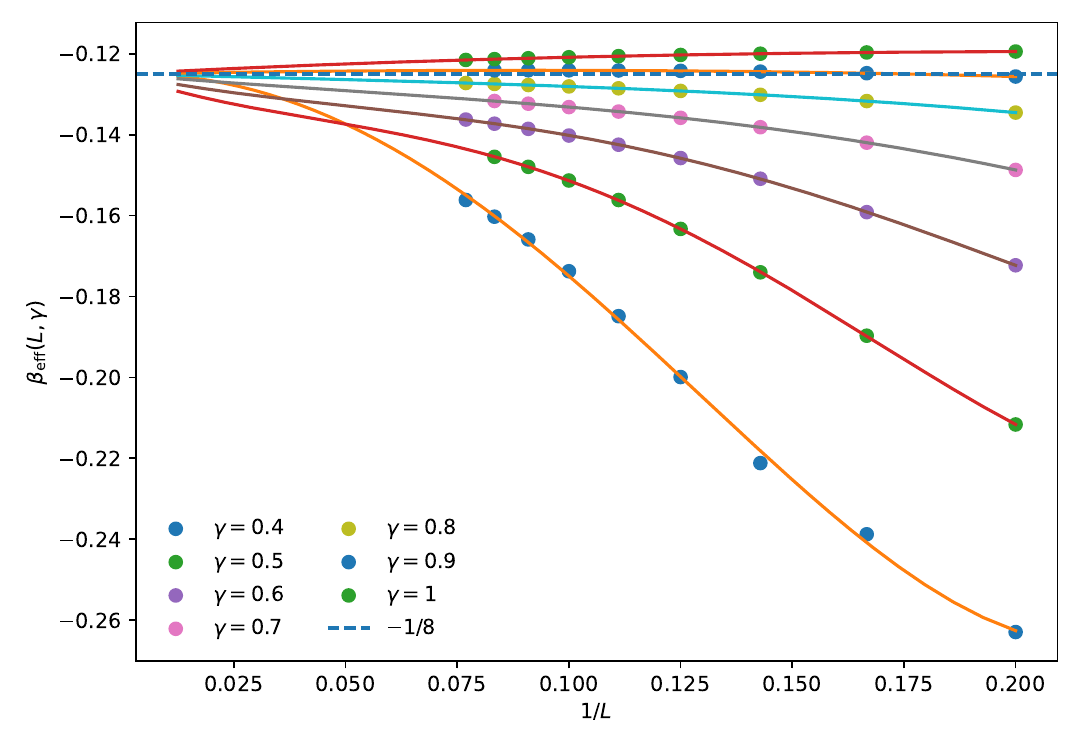}
 \caption{Complete-minor local logarithmic estimators for the generic
 open anisotropic $XY$ chain.  The displayed curves use the expected
 asymptote $-1/8$ with four inverse-size corrections.  The free-exponent
 fits are controlled for the larger anisotropies, while small $\gamma$
 exhibits a long pre-asymptotic drift and is used only as a consistency
 test of the same limiting coefficient.}
 \label{fig:xy-critical-improved}
\end{figure}

\paragraph{Nested massive response and accessible double-scaling test.}
Writing the mobile root as $z=-q$, the root equation following from
Eq.~\eqref{eq:xy-polynomial} is
\begin{equation}
 h=J_xq+\frac{J_y}{q}.
 \label{eq:xy-root-equation}
\end{equation}
Matched points with inverse localization length $m>0$ are therefore
\begin{align}
 q_{\rm top}&=\e^{-m},
 &
 h_{\rm top}(m,\gamma)&=J_x\e^{-m}+J_y\e^m,
 \\
 q_{\rm triv}&=\e^m,
 &
 h_{\rm triv}(m,\gamma)&=J_x\e^m+J_y\e^{-m}.
 \label{eq:xy-matched-improved}
\end{align}
Throughout the retained parameter window the second root remains
strictly inside the unit disk.  Hence the mobile root crossing changes
the BDI index by one, $|\Delta\omega|=1$, while the two sides have the
same leading localization length $\xi=m^{-1}$.

The boundary response discussed in the Letter is defined by a nested
massive limit.  Let
\begin{equation}
 \Delta s_\gamma(m)
 =s_{\rm top}(m,\gamma)-s_{\rm triv}(m,\gamma)
 \label{eq:xy-bulk-difference}
\end{equation}
and define the finite-size bulk-subtracted difference
\begin{align}
 \mathcal R_{\gamma,L}(m)
 ={}&
 \ln\cS_L[h_{\rm top}(m,\gamma)]
 -\ln\cS_L[h_{\rm triv}(m,\gamma)]
 \notag\\
 &-L\Delta s_\gamma(m).
 \label{eq:xy-fixed-m-response}
\end{align}
The desired response is
\begin{equation}
 \mathcal R_{1/2,\gamma}^{\rm SPT}
 =
 \lim_{m\to0^+}\lim_{L\to\infty}
 \mathcal R_{\gamma,L}(m).
 \label{eq:xy-response-final}
\end{equation}
The inner limit is taken at fixed $m>0$, so that $mL=L/\xi\to\infty$
and the two physical boundaries decouple before the critical limit is
taken.

Complete absolute-minor enumeration reaches only $L\le12$ for the
off-critical grid.  It therefore cannot simultaneously realize
$m\ll1$ and $mL\gg1$, and does not directly evaluate
Eq.~\eqref{eq:xy-response-final}.  Instead, the generic-$XY$ data test
the associated near-critical double-scaling regime by setting
$m=x/L$ and taking $L\to\infty$ at fixed $x$.  Define
\begin{equation}
 \Delta_\gamma^{(L)}(x)
 =
 \ln\cS_L[h_{\rm top}(x/L,\gamma)]
 -
 \ln\cS_L[h_{\rm triv}(x/L,\gamma)].
 \label{eq:xy-Delta-improved}
\end{equation}
This is a crossover quantity, not the fixed-$m$ thermodynamic boundary
constant.

Unlike TFI, a generic $XY$ chain is not bulk self-dual.  The paired
fields are exchanged by $m\mapsto-m$, so the bulk-density difference is
odd in $m$ and has the expansion
\begin{equation}
 L\Delta s_\gamma(x/L)
 =v_\gamma^{\rm bulk}x+O(x^3/L^2).
 \label{eq:xy-bulk-double-scaling}
\end{equation}
We subtract the exact TFI finite-size crossover and fit
\begin{equation}
 Y_\gamma^{(L)}(x)
 \equiv
 \Delta_\gamma^{(L)}(x)-\Delta_1^{(L)}(x)
 =
 v_\gamma x+
 \sum_{p=1}^P\frac1{L^p}
 \sum_{r=0}^2 c_{pr}^{(\gamma)}x^{2r+1}.
 \label{eq:xy-relative-crossover-fit}
\end{equation}
Because only double-scaling data are available, the fitted coefficient
$v_\gamma$ cannot independently separate the bulk coefficient
$v_\gamma^{\rm bulk}$ from a possible anisotropy-dependent boundary
mismatch proportional to $x$.  It should therefore be regarded as an
effective linear coefficient.  The fit tests whether an additional
\emph{nonlinear} $O(1)$ crossover function is resolved after this
linear ambiguity is removed.

The fit uses the complete $L=6,\ldots,12$ data at
$x=0.5,1,1.5,2$.  The results are
\begin{center}
\begin{tabular}{c|c|ccc}
\toprule
$\gamma$ & $P$ & $v_\gamma$ & RMS residual & maximum residual\\
\midrule
0.4 & 2 & 0.15968973 & $1.43\times10^{-4}$ & $4.96\times10^{-4}$ \\
0.4 & 3 & 0.14816710 & $6.69\times10^{-5}$ & $1.86\times10^{-4}$ \\
0.6 & 2 & 0.08568863 & $7.44\times10^{-5}$ & $1.38\times10^{-4}$ \\
0.6 & 3 & 0.08500875 & $6.96\times10^{-6}$ & $1.94\times10^{-5}$ \\
0.8 & 2 & 0.03882613 & $1.87\times10^{-5}$ & $3.91\times10^{-5}$ \\
0.8 & 3 & 0.03938950 & $9.95\times10^{-7}$ & $2.50\times10^{-6}$ \\
\bottomrule
\end{tabular}
\end{center}
The effective linear coefficient is stable for $\gamma=0.6$ and $0.8$
and shows a larger fitting-window dependence at $\gamma=0.4$, where
the available sizes are more strongly pre-asymptotic.  Independent
fixed-$x$ extrapolations are consistent with
\begin{equation}
 Y_\gamma^{(\infty)}(x)
 =v_\gamma x+B_\gamma(x),
 \qquad
 |B_\gamma(x)|\lesssim2.1\times10^{-3}
 \quad(0.5\le x\le2),
 \label{eq:xy-same-crossover}
\end{equation}
within the numerical resolution of the retained grid.  Allowing
$B_\gamma(x)=d_1x^3+d_2x^5$ does not produce a stable nonzero function
under changes of $P$ and the fitting window.  This means that no
additional nonlinear crossover is resolved over the accessible
interval; it does not prove that $B_\gamma(x)$ vanishes identically,
exclude an additional linear boundary term absorbed into $v_\gamma$,
or constrain the large-$x$ regime.

\begin{figure}[H]
 \centering
 \includegraphics[width=0.82\linewidth]
 {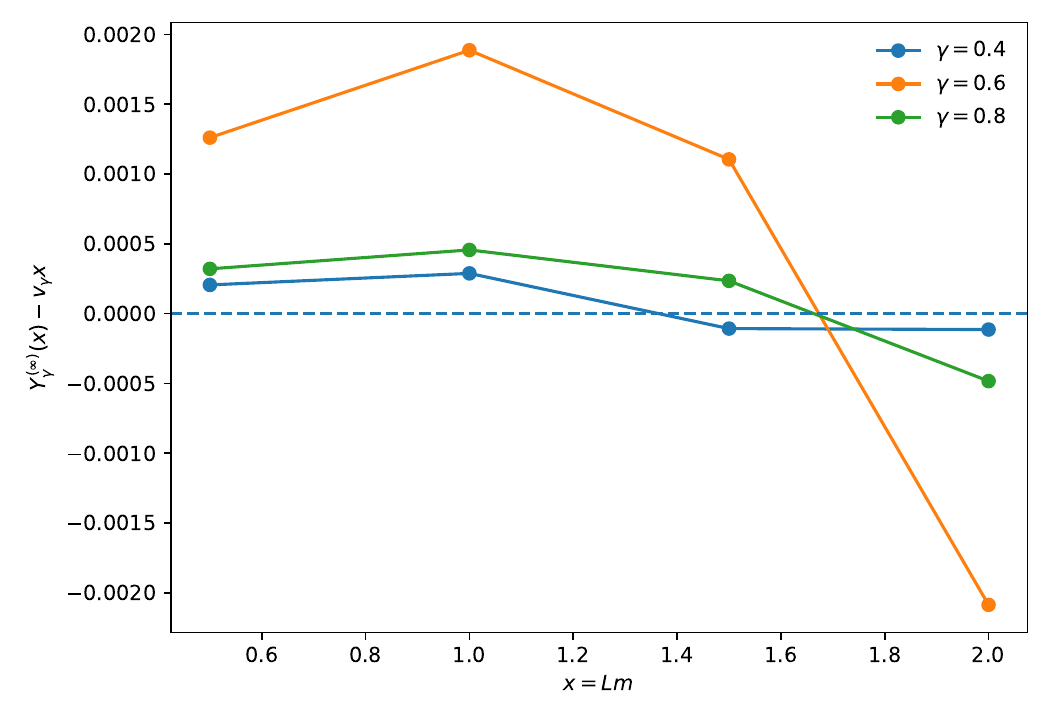}
 \caption{Residual obtained by extrapolating
 $Y_\gamma^{(L)}(x)$ to $L=\infty$ at fixed $x$ and removing its
 best-fit effective linear term $v_\gamma x$.  Within the exact-enumeration
 window $0.5\le x\le2$, no additional nonlinear crossover is resolved
 relative to TFI.  The calculation is a fixed-$x$ double-scaling test
 and does not directly implement the nested fixed-$m$ limit in
 Eq.~\eqref{eq:xy-response-final}.}
 \label{fig:xy-crossover-improved}
\end{figure}

The generic-$XY$ calculation therefore supplies two independent checks:
the critical data are consistent with the expected $-1/8$ logarithmic
coefficient, and the near-critical fixed-$x$ data reveal no nonlinear
boundary mismatch relative to TFI within the accessible interval.  It
does not independently determine the massive response in
Eq.~\eqref{eq:xy-response-final} or its large-$x$ plateau.  When combined
with the independently controlled TFI result, the absence of a resolved
finite-$x$ nonlinear mismatch is consistent with the proposed
one-channel value $\ln2$, but it should be interpreted only as an
additional non-Pfaffian crossover consistency test.

\ssection{sec:scope}{Scope, interpretation, and remaining mathematical questions}

\ssubsection{sec:two-structures}{Decimation and Pfaffian solvability are independent}

The distinction underlying the classification in Sec.~\ref{sec:model-status} is:
\begin{enumerate}[leftmargin=2.2em]
 \item An exact decomposition such as $f(z)=g(z^d)$ factorizes $\mathcal D_\beta$ and therefore holds for arbitrary $\alpha$.  It does not use the minor-summation Pfaffian.
 \item A single Pfaffian at $\alpha=1/2$ requires a coherent sign for every minor.  Free-fermion solvability, topology, or decimation alone does not guarantee this sign property.
\end{enumerate}
For the generic anisotropic $XY$ matrix in Eq.~\eqref{eq:XY-Z}, competing directed paths produce parameter-dependent cancellations in its minors.  A uniform fixed selector has not been established when $J_xJ_y\ne0$.  At the Ising endpoints, one directional coupling vanishes and the matrix returns to the bidiagonal path class.  Equal-amplitude isolated chiral points may have separate reductions, but these do not provide an off-critical formula for the generic $XY$ family.

\ssubsection{sec:physical}{Physical meaning of the response}

The total entropy is dominated by $Ls(h)$, and neither massive boundary constant is separately quantized.  A local termination change or an appended stabilizer degree of freedom can shift an individual $O(1)$ term.  Universality is observed only after: (i) extracting the two bulk densities independently, (ii) comparing the same microscopic termination on the two sides, and (iii) approaching their common critical theory.

For $z^r(h+z^d)$, the $r$ roots pinned at the origin survive on both sides.  They label the background critical boundary sector and cancel from the two-sided response.  The $d$ mobile roots cross the unit circle and give $|\Delta\omega|=d$.  One crossing channel transfers one real zero mode to each physical end.  The two boundary Majoranas form one nonlocal complex fermion with a two-dimensional occupation space, which supplies the natural interpretation of the relative factor two and hence $\ln2$.  The independent-end test in Sec.~\ref{sec:end-factorization} further shows that the corresponding $O(1)$ terms become additive for the nonsingular symmetry-preserving terminations considered here.  For reflection-related ends, this is consistent with an effective factor $\sqrt2$ from each topology-changing Majorana end, although the present analysis does not identify these factors with boundary-CFT $g$ factors.  This is an interpretation of the Pauli-participation free energy, not a thermodynamic entropy derivation.

The response is distinct from subsystem topological-magic constructions after non-Clifford doping~\cite{Nehra2025}.  Those probe nonlocally distributed bosonic SPT resources; Eq.~\eqref{eq:main-response} probes the fermionic BDI index of the undoped Jordan--Wigner chain.  The total SRE need not be an SPT invariant for its bulk-subtracted OBC response to retain symmetry-protected boundary information.

\ssubsection{sec:renyi-scope}{Exact R\'enyi-index constraints and open massive response}

For arbitrary index it is useful to distinguish the full SRE response
from the reduced $\alpha=1/2$ normalization used in the Letter.  At
paired massive points, write
\begin{equation}
 M_\alpha^a(L,m)=L s_\alpha^a(m)+C_\alpha^a(m)+o(1),
 \qquad a\in\{\mathrm{top},\mathrm{triv}\},
 \label{eq:general-alpha-expansion}
\end{equation}
and define the strict nested-limit response
\begin{equation}
 R_\alpha=
 \lim_{m\to0^+}\lim_{L\to\infty}
 \left[
 M_\alpha^{\rm top}(L,m)-M_\alpha^{\rm triv}(L,m)
 -L\Delta s_\alpha(m)
 \right],
 \label{eq:general-alpha-response}
\end{equation}
where
$\Delta s_\alpha=s_\alpha^{\rm top}-s_\alpha^{\rm triv}$ is retained
independently at every fixed mass.  The order in
Eq.~\eqref{eq:general-alpha-response} is essential: fixed $mL$ defines
a crossover trajectory rather than the massive boundary constant.
At $\alpha=1/2$, Eq.~\eqref{eq:M-half} gives
\begin{equation}
 R_{1/2}=2\cR_{1/2}
 \longrightarrow 2|\Delta\omega|\ln2,
 \label{eq:full-half-response}
\end{equation}
so the reduced response of the Letter is one half of the full
SRE---and, by Eq.~\eqref{eq:stabilizer-shannon}, one half of the
occupation-basis Shannon--R\'enyi---boundary difference.  At the
opposite endpoint, Eqs.~\eqref{eq:pmax} and
\eqref{eq:min-entropy} imply
\begin{equation}
 R_\infty=0
 \label{eq:infty-response}
\end{equation}
exactly for equal-length chains with the same finite-size convention.

The critical TFI/SSH participation problem has a boundary R\'enyi
transition at $\alpha=4$~\cite{Stephan2011,Tarighi2022,
RamirezRajabpour2025}.  In the full-SRE normalization,
\begin{equation}
 M_\alpha(L,h_c)=\mu_\alpha L+b_\alpha\ln L+O(1),
 \qquad
 b_\alpha=
 \begin{cases}
 -\tfrac14,&\alpha<4,\\
 -\tfrac16,&\alpha=4,\\
 0,&\alpha>4.
 \end{cases}
 \label{eq:critical-renyi-transition}
\end{equation}
This critical boundary transition, the Letter result in
Eq.~\eqref{eq:full-half-response}, and the exact endpoint
Eq.~\eqref{eq:infty-response} show that the massive response has a
nontrivial R\'enyi-index structure.  They do not determine the strict nested-limit function
$R_\alpha$: no value is assigned here at $\alpha=4$, and no
interpolation or locked-side branch is assumed.  Determining
$R_\alpha$ away from $\alpha=1/2$ requires an independent
thermodynamic extrapolation at every fixed mass, including a fitted
bulk-density difference, before the limit $m\to0^+$ is taken.


\end{document}